\documentclass[10pt,twocolumn,romanappendices,journal]{IEEEtran}

\usepackage[T1]{fontenc}		
\usepackage[utf8]{inputenc}	

\usepackage{cite}
\usepackage{graphicx}
\usepackage{amsmath,amssymb}
\usepackage{amsfonts}

\usepackage{comment}
\usepackage{epstopdf}
\usepackage{color}

\usepackage{hyperref}
\usepackage{subcaption}
\usepackage{nicefrac}
\usepackage{placeins} 
\usepackage[dvipsnames]{xcolor}

\usepackage{microtype}

\usepackage{makecell}
\usepackage{stfloats}

\usepackage{algpseudocode}

\usepackage[ruled,vlined,linesnumbered]{algorithm2e}
\DontPrintSemicolon                    
\SetKwInput{KwInput}{Input}
\SetKwInput{KwOutput}{Output}
\SetKwProg{Proc}{Procedure}{}{}
\SetKwComment{tcp}{$\triangleright$ }{} 
\usepackage{flushend}

\newtheorem{remark}{Remark}
\newtheorem{theorem}{Theorem}
\newtheorem{lemma}{Lemma}
\newtheorem{prop}{Proposition}
\newtheorem{coro}{Corollary}
\newtheorem{definition}{Definition}

\usepackage{hyperref}
\hypersetup{
	colorlinks = true,
	linkcolor  = red,
	citecolor  = green!80!black,
	urlcolor   = black
}

\usepackage{url}

	\newcommand{\Gc}{\mathcal{G}}

	\newcommand{\Tc}{\mathcal{T}}

\definecolor{color1}{rgb}{0.00000,0.44700,0.74100}%
\colorlet{myred}{red!80!black}%
\colorlet{myblue}{color1!80!black}%
\colorlet{green}{green!80!black}%

\usepackage{bm}

 \usepackage{booktabs}
 \usepackage{tabularx}
  \usepackage{array}
 \newcolumntype{L}{>{\raggedright\arraybackslash}X}

\newenvironment{proof}{{\noindent\it Proof:}}{\hfill $\blacksquare$\par}

\begin{document}

\title{{Vector Coded Caching Multiplicatively  Boosts MU-MIMO Systems  Under Practical Considerations}}
\author{Hui Zhao  and Petros Elia \vspace{-0.2cm}
      \thanks{This work was supported in part by EURECOM’s industrial members: ORANGE, BMW, SAP, iABG, Norton LifeLock; in part by ERC-PoC Project LIGHT (Grant 101101031) and EU H2030 project CONVERGE, by the French projects EEMW4FIX (ANR) and PERSEUS (PEPR-5G);  and in part by the Huawei France funded Chair towards Future Wireless Networks.
     Part of this work was presented at the International ITG 26th Workshop on Smart Antennas and 13th Conference on Systems, Communications, and Coding (WSA\&SCC 2023)~\cite{Zhao_WSA_SCC}. This paper has been accepted for publication in \emph{IEEE Transactions on Wireless Communications}.}
	\thanks{Hui Zhao and Petros Elia are with the Communication Systems Department, EURECOM, 06410 Sophia Antipolis, France (email: hui.zhao@eurecom.fr; petros.elia@eurecom.fr). 
	}
}


\maketitle

\begin{abstract}
This work presents a first comprehensive analysis of the impact of vector coded caching (VCC) in multi-user  multiple-input multiple-output (MU-MIMO) systems with multiple receive antennas and variable pathloss --- two key factors that critically influence systems with inherent MU unicasting behavior. We investigate two widely adopted precoding strategies: (i) block-diagonalization (BD) at the transmitter combined with maximal ratio combining (MRC) at the receivers, and (ii) zero-forcing (ZF) precoding. Our analysis explicitly accounts for practical considerations such as channel fading, channel state information (CSI) acquisition overhead, and fairness-oriented power allocation.

Our contributions span both analytical and simulation-based fronts. On the analytical side, we derive analytical expressions for the achievable throughput under BD-MRC and ZF, highlighting the performance benefits of equipping multi-antenna users with cache-aided interference management. Specifically, we develop a low-complexity BD-MRC optimization method that leverages matrix structure to significantly reduce the dimensionality involved in precoding computation, followed by solving the associated max-min fairness problem through an efficient one-dimensional search. In the massive MIMO regime, an asymptotic expression for the achievable throughput over Rayleigh fading channels is also derived.
Simulations validate our theoretical results, confirming that VCC delivers substantial performance gains over optimized cacheless MU-MIMO systems. For example, with 32 transmit antennas and 2 receive antennas per user, VCC yields throughput improvements exceeding 300\%. These gains are further amplified under imperfect CSI at the transmitter, where VCC’s ability to offload interference mitigation to the receivers ensures robust performance even in the face of degraded CSI quality and elevated acquisition costs.
\end{abstract}

	\IEEEpeerreviewmaketitle

\section{Introduction}\label{Intro_sec}

Coded caching was first introduced in the context of the single-stream Broadcast Channel (BC) in~\cite{Ali}, where $K$ cache-aided receivers request different files from a shared content library. This approach exploits each user's ability to store a fraction $\gamma \in [0,1]$ of the library, theoretically achieving $K \gamma + 1$ degrees of freedom (DoF), often referred to as the coded caching gain \cite{9163148,7959863,Lampiris_Feedback}, and it roughly corresponds to the number of users served at a time. However, practical implementations faced challenges due to the combinatorial nature of the scheme in~\cite{Ali}, which requires each file to be divided into ${K \choose K\gamma}$ non-overlapping subfiles. As a result, finite file sizes impose a fundamental constraint, reducing the real DoF to $\Lambda \gamma +1$, where \( \Lambda \ (\Lambda \ll K)\), constrained by the file size, denotes the number of distinct cached contents --- equivalently referred to as \emph{cache states} --- that can be stored at the users. In practice, these constraints frequently limit achievable DoF to small single-digit values \cite{Yan}.  

To overcome these limitations, research has explored the integration of coded caching with multi-antenna transmission, leading to the concept of \emph{multi-antenna coded caching}. Notable related contributions include the works of Shariatpanahi et al. \cite{Shariatpanahi} and Naderializadeh et al. \cite{Naderializadeh}, as well as studies on physical-layer (PHY) multi-antenna coded caching \cite{Shariatpanahi_TIT2019}. Other investigations have examined the scalability of content delivery rates in massive multiple-input multiple-output (MIMO) and cell-free networks \cite{Ngo,Bayat}, while recent efforts have focused on cache-aided precoders for mitigating inter-stream interference \cite{Tolli_TWC} and location-aware caching for wireless extended reality applications \cite{Mahmoodi_location}. Further contributions can be found in \cite{9163148,Lampiris_Feedback,Bergel2018,Meixia_Tao_TWC2019,MohammadJavadTWCwithus2021,Antti9348098,Antti9083779} and related works.  

Despite these advancements, the early multi-antenna coded caching strategies remained constrained by file-size limitations, and as a consequence, their effective DoF remained upper-bounded by $\Lambda \gamma + Q$ \cite{Lampiris_Feedback}, where $Q$ represents the multiplexing gain enabled by multiple antennas. This limitation simply meant that the presented caching gains were dwarfed by conventional multiplexing gains, particularly in massive MIMO where $\Lambda \gamma \ll Q$, as extensively validated in field trials \cite{report16layers}.

This limitation was overcome with \emph{vector coded caching (VCC)} in \cite{Lampiris_JSAC}, whose introduction significantly enhanced the role of caching in multi-antenna systems. Unlike previous XOR-based approaches, VCC does not rely on multicasting; instead, it employs a clique-based structure on vectors, achieving a real DoF of $Q (\Lambda \gamma +1)$ under file-size constraints. This represents \emph{a multiplicative DoF gain} over the cacheless case (DoF $Q$), surpassing the additive improvement of previous XOR-based methods (real DoF $\Lambda \gamma + Q$). Notably, this performance enhancement comes with a manageable subpacketization cost. Recently, VCC, also referred to as \emph{signal-level} coded caching~\cite{Bakhshzad}, has evolved into a broad class of schemes extensively investigated from an information-theoretic perspective, typically focusing on the delivery-load--subpacketization tradeoff (e.g.,~\cite{Emanuele,Elizabath,Ting}).  Moving beyond the infinite-SNR DoF analysis, the study in~\cite{Bakhshzad} examined the delivery performance of VCC at finite SNRs over wireless fading channels through numerical optimization, incorporating user locations into the cache placement phase.
Another line of research in \cite{Zhao_VectorCC} provided a first statistical analysis of VCC under practical SNR regimes, focusing on the case of single-antenna receivers as well as symmetric Rayleigh fading channels, incorporating beamforming gains and CSI acquisition costs, thus providing early evidence of the significant potential of VCC in wireless systems, demonstrating a substantial multiplicative improvement in spectral efficiency over conventional (cacheless) multi-user (MU) multiple-input single-output (MISO) systems.

\subsection{Limitations of Existing Works and Technical Challenges}

Going beyond the preliminary study in~\cite{Zhao_VectorCC}, our work here examines for the first time critical factors pertaining to a more realistic implementation of VCC in wireless systems, including the impact of channel asymmetry due to pathloss, the emergence of multi-antenna receivers, and the need for power allocation. In particular, channel asymmetry among users is crucial due to the near-far bottleneck associated to coded caching techniques, where this bottleneck has been shown to significantly reduce spectral efficiency in finite-SNR regimes \cite{Mingyue_Ji,Tegin2020,Zhao2021_ISIT,Zhao_TWC}. Similarly, the growing adoption of multi-antenna receivers \cite{Antti_ICASSP,Antonious,Youlong_Cao} necessitates an investigation into how they enhance VCC performance, particularly under heterogeneous pathloss conditions.  
These practical considerations introduce new challenges, such as determining the optimal multiplexing gain $Q^\star$, and thus balancing the intricate tradeoff between multiplexing and beamforming gains as well as CSI overhead~\cite{Zhao_VectorCC}. Moreover, multi-antenna receiver combining introduces a complex interaction between transmit precoding and receiver processing \cite{Tehrani2025Asymmetric, Tehrani2024CacheAided, 10694609}.  

This same channel asymmetry studied here, now also raises concerns regarding user fairness, typically addressed via max-min fairness (MMF) techniques \cite{Tolli_TWC, Bakhshzad}, which optimize power allocation to maximize the minimum transmission rate. In VCC, analyzing such fairness aspects is automatically more involved, with optimizations that tend to be even more complex due to the potentially larger number of users involved (MMF leads to a non-convex, NP-hard problem \cite{Sadeghi_TWC2018}), leading to potentially higher computational complexity, thus further complicating resource allocation.  

Together, these challenges --- including pathloss variations, heterogeneous receiver antenna configurations, multiplexing-beamforming tradeoffs, and MMF constraints --- make joint optimization particularly difficult. Even with techniques like block diagonalization (BD) precoding and maximal ratio combining (MRC) \cite{Haardt_BD,caire2010multiuser,7335586,10634195}, deriving closed-form spectral efficiency expressions remains extremely challenging. These challenges are tackled here, in the presence of cache-aided interference management. 

\vspace{-0.3cm}
\subsection{Contributions and Paper Structure}
Our work captures the above crucial ingredients of fading, pathloss, fairness-based power allocation, CSI overhead, linear precoding, and receiver combining, offering a first-ever comprehensive study of VCC in practical MU-MIMO settings. Given that Video-on-Demand (VoD) traffic accounts for over 70\% of total network data consumption \cite{Ciso_forest}, efficient caching-based delivery mechanisms are essential in alleviating network congestion. By addressing these challenges, our study provides a systematic framework for enhancing spectral efficiency in modern wireless networks through VCC. Our main contributions are as follows:

\begin{enumerate}

\item We provide a rigorous analytical framework for VCC under BD-MRC, focusing on both precoding and receiver-side processing. Specifically, we leverage the idempotent and Hermitian properties of projection matrices to design a low-complexity BD precoder. Notably, for a given transmitted symbol’s precoding column, this approach allows us to replace the eigenvalue decomposition of an \(L \times L\) matrix (where \(L\) is the number of base-station (BS) antennas) with the decomposition of a matrix whose dimension is determined solely by the number of receive antennas at the target user. Since the BS typically has far more antennas than any individual user, this substantially reduces computational complexity. Furthermore, we derive a theoretical expression for the transmission rate of a specific user and subsequently optimize the power allocation across multiple symbols intended for that user to maximize the achievable transmission rate under the constraint of the user's total allocated transmission power.

\item We then analytically determine the optimal effective sum-rate (cf.  Definition \ref{sum_rate_def}) of the BD-MRC based VCC under the MMF requirement. To address the NP-hard complexity of MMF optimization, we transform it into a simple one-dimensional search problem where the objective equality has a unique root. 
Furthermore, we analytically derive the upper and lower bounds of the root, which substantially constrain the search range for determining the root numerically, making the root-finding process more efficient.

\item We subsequently focus on Rayleigh fading channels with variable pathloss for each user and analyze the delivery performance of the optimized BD-MRC based VCC in the massive MIMO regime.\footnote{We also consider the special case of single-antenna receivers under BD precoding and the MMF requirement, and derive the closed-form expression for the optimal sum-rate, which was rightfully identified as a hard problem in~\cite{caire2010multiuser}. We refer to \cite[Lemma 4.3]{EURECOM+7083} for more details.} Despite the requirement of a one-dimensional search, the simplified objective equality significantly reduces the search complexity. Furthermore, assuming an equal number of receiving antennas per user, we derive a simple closed-form expression for the optimal effective sum-rate. 
This simple expression not only allows us to bypass tedious numerical optimization but also crisply reflects the impact of the various factors under consideration.

\item Additionally, we analyze the performance of Zero-Forcing (ZF) precoding for multi-antenna receivers in VCC over Rayleigh fading channels with pathloss, and derive closed-form expressions for lower and upper bounds on the effective sum-rate, in the presence of MMF. Our numerical results validate these bounds as accurate approximations, demonstrating that for practical receiver configurations in VCC, ZF precoding optimized based on pathloss closely matches the performance of optimized BD-MRC.\footnote{This ZF precoding corresponds to the work of the conference version~\cite{Zhao_WSA_SCC}, which though constrains itself to a special case of an equal number of receiving antennas per user. Furthermore, the numerical comparisons with the optimized BD-MRC were also omitted in~\cite{Zhao_WSA_SCC}.}

\end{enumerate}

\emph{Notations:} For a positive integer $a$, we use $[a]$ to denote the set $\{1,2,\cdots,a\}$. $|\cdot|$ denotes either the cardinality of a set or the magnitude of a complex number. $\mathbb{C}$ denotes the set of complex numbers, while $\mathbb{E}\{\cdot\}$ and ${\rm Tr}\{\cdot\}$ represent the average operator and the trace operator respectively. ${\rm Diag}\{ a_1, \cdots, a_n \}$ denotes a diagonal matrix with  diagonal elements  $ a_1, \cdots, a_n$, while $\Psi \setminus \phi$ denotes the set $\Psi$ after removing the element $\phi$.  We will use $(x)^+ \triangleq \max(x,0)$. All vectors are column vectors. Furthermore, we use ${\bf A}^T$, ${\bf A}^*$ and ${\bf A}^H$ to denote the non-conjugate transpose, conjugation and conjugate transpose of the matrix ${\bf A}$ respectively. $[{\bf A}]_{\ell,\vartheta}$ denotes the $(\ell,\vartheta)$-th element in  ${\bf A}$. Also, ${\bf 0}_L \in \mathbb{C}^{L}$ denotes the vector with all zero elements, and ${\bf I}_L \in \mathbb{C}^{L \times L}$ represents the identity matrix. We additionally use~$X \sim \mathcal{Y}$ to denote that~$X$ follows the statistical distribution~$\mathcal{Y}$, and we use $\mathcal{CN}({\bf m},{\bf \Sigma})$  to denote complex Gaussian distribution with mean vector ${\bf m}$ and covariance matrix ${\bf \Sigma}$. As \(L \to \infty\), we write \(f(L) \simeq g(L)\) to indicate that \(\lim_{L \to \infty} |f(L) - g(L)| = 0\). Finally, we use $||\cdot||$ to represent the $\ell_2$-norm  operator. Moreover, the main variables used in the paper are summarized in Table~\ref{tab:notations_new}.

\emph{Paper Structure:} We introduce the system model in Section~\ref{sys_sec}. Then in Section \ref{VectorCC_Sec} we elaborate on BD-MRC based VCC, and in Section \ref{multi_ana_sec} we present its corresponding analysis. Subsequently, in Section \ref{single_ana_sec}, we mathematically analyze the performance of the ZF precoder for multi-antenna receivers. Numerical results and performance comparisons are presented in Section \ref{numerical_sec}, while Section \ref{conclude_sec} concludes the paper.

\begin{table}[!t]
  \caption{Notations of Important Variables}
  \label{tab:notations_new}
  \centering
  \footnotesize
  \setlength{\tabcolsep}{4pt}      
  \renewcommand{\arraystretch}{1.15}
  \begin{tabularx}{0.98\linewidth}{@{}lL@{}} 
    \toprule
    \textbf{Notation} & \textbf{Definition} \\
    \midrule
    $K$  & Total number of users \\
    $L$ & Number of transmit antennas \\
    $P_{\rm tot}$ &  Maximum allowable transmit power \\
    $F$  & Total information bits per file\\
    $\mathcal{F}$ & Library with each file of $F$ bits\\
    $N$ &Total number of files in the library $\mathcal{F}$\\
    $W_n$ & The $n$-th file of library $\mathcal{F}$\\
    $W_n^\mathcal{T}$ & Subfile of $W_n$ labeled by  $\Lambda\gamma$-tuple $\mathcal{T}\subseteq[\Lambda]$ \\
    $\gamma $ & Normalized cache size relative to the library content\\
    $\Lambda$ & Number of distinct cache states\\
    $Q$ & Number of users for signal decoding in a user-group of VCC\\ 
    $Q'$ & Number of users served at a time in cacheless MU-MIMO\\
    $G$  & $=\Lambda \gamma +1$, nominal gain of coded caching\\
    $\mathcal{G}$ & Effective gain over cacheless MU-MIMO\\
    $\mathcal{G}^\star$  & Effective gain with $Q'$ and $Q$ being independently optimized\\
    $\Psi$   & Set of $G$ user groups, each group with a distinct cache state    \\
    $B$ & Number of users caching the same content\\
    $\beta_{\psi,k}$ &Large-scale fading and/or pathloss to user ${\rm U}_{\psi,k}$\\
    $d_{\psi,k}$ & File index required by user ${\rm U}_{\psi,k}$\\
    $M_{\psi,k}$ & Number of receive antennas at user ${\rm U}_{\psi,k}$\\
    $J_{\psi,k}$ & Number of signal symbols simultaneously sent to user ${\rm U}_{\psi,k}$\\
    $M_\psi$ & Total number of receive antennas in user-group $\psi$\\
    $J_\psi$ & Total number of signal symbols  sent to user-group $\psi$\\
    $s_{\psi,k,q}$ & $q$-th data symbol sent to user ${\rm U}_{\psi,k}$\\
    ${\bf s}_{\psi,k}$ & Data vector sent to user ${\rm U}_{\psi,k}$\\
    ${\bf P}_{\psi,k}$ & Diagonal matrix for power allocation to user ${\rm U}_{\psi,k}$\\
    $P_{\psi,k,q}$ & Transmit power allocated to the $q$-th data symbol to user ${\rm U}_{\psi,k}$\\
    ${\bf V}_{\psi,k}$ & Precoding matrix for user ${\rm U}_{\psi,k}$\\
    ${\bf v}_{\psi,k,q}$ & $q$-th column of  ${\bf V}_{\psi,k}$\\
    ${\bf s}_\psi$  & Data vector sent to user-group $\psi$\\
    ${\bf V}_\psi$ & Precoding matrix for user-group $\psi$\\
    ${\bf H}_{\psi,k}$ & Channel matrix from transmitter to user ${\rm U}_{\psi,k}$\\
    ${\bf H}_{\psi}$ & $ \triangleq \big[ {\bf H}_{\psi,1}, \cdots, {\bf H}_{\psi,Q}  \big] $ \\
    ${\bf H}_{\psi,-k}$ & $  \triangleq  [ {\bf H}_{\psi,1}, \cdots,  {\bf H}_{\psi,k-1},  {\bf H}_{\psi,k+1}, \cdots,  {\bf H}_{\psi,Q}  ] $ \\
    ${\bf T}_{\psi,-k}$ & Projection matrix of the null-space of  ${\bf H}_{\psi,-k}^*$\\
    ${\bf R}_{\psi,k}$ & Receiver-side matrix for signal combining at user ${\rm U}_{\psi,k}$\\
    ${\bf r}_{\psi,k,q}$ & $q$-th column of  ${\bf R}_{\psi,k}$ \\
    ${\bf z}_{\psi,k}$ &  AWGN at user ${\rm U}_{\psi,k}$ and ${\bf z}_{\psi,k} \sim \mathcal{CN}({\bf 0}_{M_{\psi,k}},N_0 {\bf I}_{M_{\psi,k}})$ \\
    $\xi_{G,Q}$ & CSI costs of simultaneously serving $G Q$ users in VCC\\
    $\lambda_{\psi,k,q}$ & $q$-th largest (non-zero) eigenvalue of ${\bf H}_{\psi,k}^T {\bf T}_{\psi,-k} {\bf H}_{\psi,k}^*$\\
    ${\bf t}_{\psi,k,q}$  &Eigenvector associated with the eigenvalue $\lambda_{\psi,k,q}$\\
    $\lambda_{\psi,k}^{\min} $ & $\triangleq \min_{q \in [J_{\psi,k}]}\{\lambda_{\psi,k,q}\}$ \\
    $\lambda_{\psi,k}^{\max}$ & $\triangleq \max_{q \in [J_{\psi,k}]}\{\lambda_{\psi,k,q}\}$\\
    \bottomrule
  \end{tabularx}\vspace{-0.2 cm}
\end{table}

\section{System Model}\label{sys_sec}

We consider a single-cell downlink MU-MIMO system, where a BS equipped with $L$ antennas serves $K$ cache-aided multi-antenna users, each requesting a different file from a file library $\mathcal{F}$ of $N \ (N \ge K)$ equal-sized files. 
The BS has full access to the library $\mathcal{F}$, while each user can only cache a fraction $\gamma \in [0,1]$ of the library content, which is to be done during off-peak hours. The scheme is presented in detail in  Section \ref{JSAC_Intro}, and it directly builds on the VCC principles from~\cite{Lampiris_JSAC}, adapted here for the setting of multi-antenna receivers. Due to the aforementioned finite file-size constraints, we consider $\Lambda \ (\Lambda \leq K)$ different cache states, forcing $B = \frac{K}{\Lambda}$ users to have the same cache content (as elaborated in~\cite{Lampiris_JSAC}). Again as elaborated in~\cite{Lampiris_JSAC}, during the subsequent delivery phase, we consider $G \triangleq \Lambda \gamma +1$ user groups (each group consisting of users with the same cache contents) that are  selected for service in each transmission round, and from each group, we select $Q \in [B]$ users that will receive information during that same round.  This indicates that there are as many as  $G Q$ users served at a time. 

We use $\Psi$ to denote the $G $ groups selected for service during a specific transmission round. 
We also use ${\rm U}_{\psi,k}$ to denote the $k$-th ($k \in [Q]$) active user in group $\psi \in \Psi$. 
We further assume that each served user ${\rm U}_{\psi,k}$ is equipped with $M_{\psi,k}$ receive antennas. We use $M_\psi \triangleq \sum_{k \in [Q]} M_{\psi,k}$ to denote the total number of active receive antennas in group $\psi \in \Psi$. As a result of having multi-antenna receivers, the BS can simultaneously communicate a certain number $J_{\psi,k}$ of symbols to ${\rm U}_{\psi,k}$, thereby enhancing the throughput to ${\rm U}_{\psi,k}$, where the maximum achievable value of $J_{\psi,k}$ will naturally depend on $L$ and $M_{\psi,k}$, as well as will depend on the adopted precoding scheme. We further use 
\begin{align}
&{\bf s}_{\psi,k} \triangleq [s_{\psi,k,1}, \cdots,  s_{\psi,k,J_{\psi,k}}]^T  \in \mathbb{C}^{J_{\psi,k}}, \notag\\
&{\bf P}_{\psi,k} \triangleq {\rm Diag} \Big\{\sqrt{P_{\psi,k,1}}, \cdots, \sqrt{P_{\psi,k,J_{\psi,k}}} \Big\} \in \mathbb{C}^{J_{\psi,k}\times J_{\psi,k}}, \notag
\end{align}
to denote the data vector to ${\rm U}_{\psi,k}$ and the corresponding power allocation matrix respectively, where the independent variables $\{s_{\psi,k,q}: {q \in [J_{\psi,k}]}\}$ have zero-mean and unit-power. As we can now see, $J_\psi \triangleq \sum_{k \in [Q]} J_{\psi,k}$ corresponds to the total number of data symbols simultaneously sent to group $\psi$. We also note that each signal vector ${\bf s}_{\psi,k}$ destined for ${\rm U}_{\psi,k}$ will be  pre-processed into ${\bf V}_{\psi,k} {\bf P}_{\psi,k} {\bf s}_{\psi,k}$ by a precoding matrix 
\begin{align}
{\bf V}_{\psi,k} \triangleq [ {\bf v}_{\psi,k,1}, \cdots, {\bf v}_{\psi,k,J_{\psi,k}} ] \in \mathbb{C}^{L \times J_{\psi,k}},  \notag
\end{align}
where each unit-norm vector ${\bf v}_{\psi,k,q} \in \mathbb{C}^{L}$ precodes $s_{\psi,k,q}$. 

Let us define 
\begin{align}
&{\bf s}_\psi \triangleq [{\bf s}_{\psi,1}^T, \cdots, {\bf s}_{\psi,Q}^T   ]^T \in \mathbb{C}^{J_\psi},  \notag\\
&{\bf V}_{\psi} \triangleq \big[ {\bf V}_{\psi,1}, \cdots, {\bf V}_{\psi,Q} \big] \in \mathbb{C}^{L \times J_\psi}, \notag
\end{align}
where ${\bf s}_\psi$ denotes the vector of data destined for the $Q$ served users in group $\psi$, and where ${\bf V}_{\psi}$ represents the corresponding precoding scheme. The transmitted signal ${\bf x}_\Psi \in \mathbb{C}^{L}$ for the groups in $\Psi$ then takes the form
\begin{align}\label{signal_transmit_multi}
    {\bf x}_\Psi \!=\! \sum\nolimits_{\psi \in \Psi} \! \sum\nolimits_{k \in [Q]} \!\! {\bf V}_{\psi,k}   {\bf P}_{\psi,k} {\bf s}_{\psi,k} \!=\! \sum\nolimits_{\psi \in \Psi} \!\! {\bf V}_{\psi}   {\bf P}_{\psi} {\bf s}_{\psi}
\end{align}
where the diagonal matrix ${\bf P}_\psi \in \mathbb{C}^{J_\psi\times J_\psi}$ contains the ordered diagonal elements of matrices $\{ {\bf P}_{\psi,k}:  k \in [Q] \}$, which is responsible for allocating power to the symbols in ${\bf s}_\psi$.

\begin{remark}\label{Feature_VecterCC}
    This so-called VCC approach, simply `collapses' (by linearly combining) a carefully selected set of $|\Psi|=G$ vectors into the single vector in~\eqref{signal_transmit_multi} for one-shot transmission. Without caching, this would require $G$ sequential transmissions.
\end{remark}

Given ${\bf x}_\Psi$ in \eqref{signal_transmit_multi}, the received signal vector ${\bf y}_{\psi,k} \in \mathbb{C}^{M_{\psi,k}}$ at user ${\rm U}_{\psi,k}$ takes the form of \eqref{received_sig_Multi},  shown at the top of this page,
\begin{figure*}
\begin{align}\label{received_sig_Multi}
    {\bf y}_{\psi,k} =   {\bf H}_{\psi,k}^T  {\bf V}_{\psi,k} {\bf P}_{\psi,k} {\bf s}_{\psi,k}   + \underbrace{ {\bf H}_{\psi,k}^T \sum\nolimits_{k' \in [Q] \setminus k}    {\bf V}_{\psi,k'} {\bf P}_{\psi,k'} {\bf s}_{\psi,k'} }_{\text{intra-group interference}}
      + \underbrace{ {\bf H}_{\psi,k}^T  \sum\nolimits_{\phi \in \Psi \setminus \psi} \sum\nolimits_{\vartheta \in [Q]}  {\bf V}_{\phi,\vartheta} {\bf P}_{\phi,\vartheta} {\bf s}_{\phi,\vartheta} }_{\text{inter-group interference}} + {\bf z}_{\psi,k}
\end{align}
\hrulefill\vspace{-0.2cm}
\end{figure*}
where ${\bf z}_{\psi,k} \sim \mathcal{CN}({\bf 0}_{M_{\psi,k}},N_0 {\bf I}_{M_{\psi,k}})$ denotes the Additive White Gaussian Noise (AWGN), and where  ${\bf H}_{\psi,k} \in \mathbb{C}^{L \times M_{\psi,k} }$ denotes the channel matrix from the BS to ${\rm U}_{\psi,k}$.  
As each cache-aided user ${\rm U}_{\psi,k}$ knows --- as we detail in Section \ref{JSAC_Intro} --- the messages $\{{\bf s}_{\phi,\vartheta}: {\phi \in \Psi \setminus \psi, \vartheta \in [Q]}\}$ intended by the active users in other cache states in $\Psi$, we can conclude that the inter-group interference in \eqref{received_sig_Multi} can be removed by using the cached content in ${\rm U}_{\psi,k}$ and the composite CSI $\{{\bf H}_{\psi,k}^T {\bf V}_{\phi,\vartheta} {\bf P}_{\phi,\vartheta}: {\phi \in \Psi \setminus \psi, \vartheta \in [Q]}\}$, the cost of which we will account for in our analysis. 

Decoding at each user ${\rm U}_{\psi,k}$ will involve the removal of the known cache-aided interference, and the subsequent processing of the remaining received signal by a receiver-side matrix ${\bf R}_{\psi,k} \triangleq [ {\bf r}_{\psi,k,1}, \cdots, {\bf r}_{\psi,k,J_{\psi,k}}] \in \mathbb{C}^{  M_{\psi,k} \times J_{\psi,k} }$, which is channel-dependent and whose columns are power-normalized to 1. 
Thus, in the end, after employing cache content to remove inter-group interference, and after applying the decoding matrix ${\bf R}_{\psi,k}$, the signal vector for decoding at ${\rm U}_{\psi,k}$ takes the form
\begin{align}\label{receive_signal_multi}
    {\bf y}_{\psi,k}'    = &   {\bf R}_{\psi,k}^H {\bf H}_{\psi,k}^T  {\bf V}_{\psi,k} {\bf P}_{\psi,k} {\bf s}_{\psi,k} + {\bf z}_{\psi,k}' \notag\\
    &+ {\bf R}_{\psi,k}^H  \sum\nolimits_{k' \in [Q] \setminus k}   {\bf H}_{\psi,k}^T   {\bf V}_{\psi,k'} {\bf P}_{\psi,k'}  {\bf s}_{\psi,k'} .
\end{align}
It is easy to see that ${\bf z}_{\psi,k}' \triangleq {\bf R}_{\psi,k}^H {\bf \bf z}_{\psi,k} \in \mathbb{C}^{J_{\psi,k}} $ still follows a multi-variate complex Gaussian distribution.

\section{BD-MRC Based Vector Coded Caching}\label{VectorCC_Sec}

In this section, we will first elaborate on VCC for multi-antenna receivers, and then proceed to design the coherent BD-MRC scheme. Finally, we will present the main performance metrics for analysis in this paper.

\subsection{Signal-Level Vector Coded Caching for Finite SNR}\label{JSAC_Intro}
			
We will build on the general vector-clique structure in~\cite{Lampiris_JSAC}, which effectively super-imposes $G$ differently precoded data vectors. We will here though carefully choose the precoding schemes, while also carefully calibrating the dimensionality of our super-imposed vectors by deciding how many users to serve from or groups. This will allow us to control CSI costs and to efficiently allocate power across users efficiently, thus controlling two aspects that crucially affect the performance in practical SNR regimes \cite{Zhao_VectorCC}. 

We proceed to describe the cache placement phase and the subsequent delivery phase. For clarity, we summarize the main VCC design in \textbf{Algorithm}~\ref{alg:vcc2e}.
		
\begin{algorithm}[t]
\caption{Vector Coded Caching (VCC)}
\label{alg:vcc2e}


\Proc{Placement Phase}{}{
  Partition each file $W_n$ into $\binom{\Lambda}{\Lambda\gamma}$  non-overlapping subfiles:
  $W_n \longrightarrow \{\, W_n^{\mathcal T} : \ \mathcal T\subseteq[\Lambda],~|\mathcal T|=\Lambda\gamma \,\}$\;
 Divide \(K\) users into \(\Lambda\) groups, each with \(B = K / \Lambda\) users, indexed by \([\Lambda] \triangleq \{1, 2, \ldots, \Lambda\}\).\;
  All users in the \(g\)-th ($g \in [\Lambda]$) group store identical content:
\(\mathcal{Z}_{\mathcal{D}_g}=\{\, W_n^{\mathcal{T}} : \ g\in\mathcal{T},~\forall n\in[N] \,\}\)\;
}
\textbf{End Procedure}\;

\vspace{0.5cm}

\Proc{Delivery Phase}{}{
 Each user requests a distinct file\;
  \For{$b=1$ \KwTo $B/Q$}{
    Each encoding process serves a \emph{unique} selection of \(Q \le B\) users from each group\;
    \For{stage $t=1$ \KwTo $\binom{\Lambda}{\Lambda\gamma+1}$}{
      In each stage, choose a \emph{unique} subset $\Psi\subseteq[\Lambda]$ with $|\Psi|=\Lambda\gamma+1$\;
       Denote the requested file of user \(\mathrm{U}_{\psi,k}\) by \(W_{d_{\psi,k}}\), for \(\psi \in \Psi\) and \(k \in [Q]\)\;
      Transmit $\{\, W_{d_{\psi,k}}^{\Psi\setminus\psi} : \ \psi\in\Psi,~k\in[Q] \,\}$ \; 
      Design transmit signal ${\bf x}_\Psi$ in \eqref{signal_transmit_multi} by mapping $W_{d_{\psi,k}}^{\Psi\setminus\psi} \longrightarrow {\bf s}_{\psi,k}$ \;
      For received signal ${\bf y}_{\psi,k}$ in \eqref{received_sig_Multi}, each user $\mathrm U_{\psi,k}$ cancels inter-group interference via cached content:
      $\{\, W_n^{\Psi\setminus\phi} : \ \phi\in\Psi\setminus\psi,~\forall n\in[N] \,\}$\;
      Intra-group interference in \eqref{received_sig_Multi} is mitigated via precoding and receiver processing (see \eqref{receive_signal_multi})\;
    }
   Selected \(Q\) users per group obtain requested files \;
  }
}
\textbf{End Procedure}\;
\end{algorithm}

\subsubsection{Placement Phase}
The cache placement policy is from \cite{Lampiris_JSAC}.
The first step involves the partition of each library file $W_n$ into ${\Lambda \choose \Lambda \gamma}$ non-overlapping equally-sized subfiles
$\big\{ W_n^\mathcal{T}:  \mathcal{T} \subseteq [\Lambda], |\mathcal{T}|=\Lambda \gamma  \big\}$, each labeled by some $\Lambda\gamma$-tuple $\mathcal{T}\subseteq[\Lambda]$. As discussed in Section~\ref{Intro_sec}, the number of cache states $\Lambda$ is chosen to satisfy the file-size constraint.
Subsequently the $K$ users are \emph{arbitrarily} separated into $\Lambda$ disjoint groups $\mathcal{D}_1,\mathcal{D}_2,\dots,\mathcal{D}_\Lambda$, where the $g$-th user-group $ \mathcal{D}_g \triangleq \big\{b' \Lambda +g\big\}_{b'=0}^{B-1} \subseteq [K]$ consists of $B = \frac{K}{\Lambda}$ users\footnote{We will henceforth consider $K$ to be a multiple of $\Lambda$ for the sake of clarity of exposition, without though limiting the scope of the results. The general case can be readily handled (cf.~\cite{Lampiris_JSAC}) with a slight performance loss (cf.~\cite{Zhao_VectorCC}).}. 
All the users belonging to the same group are assigned the same cache state and thus proceed to cache \emph{identical} content. In particular, for those in the $g$-th group, this content takes the form
$\mathcal{Z}_{\mathcal{D}_g} = \big\{ W_n^{\mathcal{T}}:  \Tc \ni g,\, \forall n\in[N]  \big\}$. We clarify that this grouping as well as the entire placement phase (which takes place rarely, and over off-peak hours), are naturally done before the users' requests take place.

\subsubsection{Delivery Phase}

	This phase starts when each user ${\rm U}_{\psi,\kappa}$, $\psi \in \Psi$ and $\kappa \in [B]$, simultaneously asks for its intended file, denoted here by $W_{d_{\psi,\kappa}}$, $d_{\psi,\kappa} \in [N]$. 	The BS then selects $Q\leq B$ users from each group\footnote{It is not difficult to see that in the absence of caching, corresponding to $G=1$, this $Q$ would have played the role of the multiplexing gain.}. By doing so, the BS decides to first `encode' over the first $\Lambda Q$ users, and to repeat the encoding process $B/Q$ times\footnote{To see this, let us consider the example in Fig.~\ref{Example_VectorCC}, where $G=2$ groups (2 cache states), each with $Q=2$ active users, are selected for service in each transmission stage. The algorithm that we describe here will be first applied to the first $\Lambda Q$ users, and then, after this delivery is done, the same algorithm will apply to the remaining $\Lambda Q$ users, thus eventually satisfying all $K$ users. Also note that a small amount of additional subpacketization can easily resolve the case where $B/Q$ may not be an integer~\cite{Lampiris_JSAC}.}.
	To deliver to the $\Lambda Q$ users, the BS employs ${\Lambda \choose \Lambda \gamma+1}$ sequential transmission stages. During each such stage, the BS simultaneously serves a unique set $\Psi$ of $|\Psi| = \Lambda \gamma+1$ groups, corresponding to a total of
	$Q (\Lambda \gamma+1)$ users served at a time (i.e., per transmission stage), during which all the subfiles $\{W_{d_{\psi,k}}^{\Psi \setminus \psi }: \psi \in \Psi, k \in [Q]\}$ are sent simultaneously. After mapping each such subfile $W_{d_{\psi,k}}^{\Psi \setminus \psi }$ into a complex-valued data vector ${\bf s}_{\psi,k}$, the transmission is as shown in~\eqref{signal_transmit_multi}.  We here note that  ${\rm U}_{\psi,k}$ has cached the subfiles $\{W_{n}^{\Psi \setminus \phi}: \phi \in \Psi \setminus \psi, \forall n\in[N] \}$, thereby enabling ${\rm U}_{\psi,k}$ to cancel the inter-group interference resulting from $\{W_{d_{\phi,\vartheta}}^{\Psi \setminus \phi }: \phi \in \Psi \setminus \psi,  \vartheta \in [Q]\}$. Note that the inter-group interference cancellation also requires knowledge of the composite CSI (cf. \eqref{received_sig_Multi}). On the other hand, the intra-group interference experienced at each user ${\rm U}_{\psi,k}$, resulting from $\{W_{d_{\psi,k'}}^{\Psi \setminus \psi }:   k'\in [Q] \setminus k\}$, can be handled --- as we will do here --- with linear precoding that ‘separates’ the signals of the users from the same group.
    At the end of the ${\Lambda \choose \Lambda \gamma+1}$ transmission stages, all the $\Lambda Q$ users obtain their intended files. By repeating this process $B/Q$ times per group, across all groups, all the $K$ users obtain all their intended files. We refer to \cite{Lampiris_JSAC} for more details.

\begin{figure}[!t]
\centering
\includegraphics[width= 3.4 in]{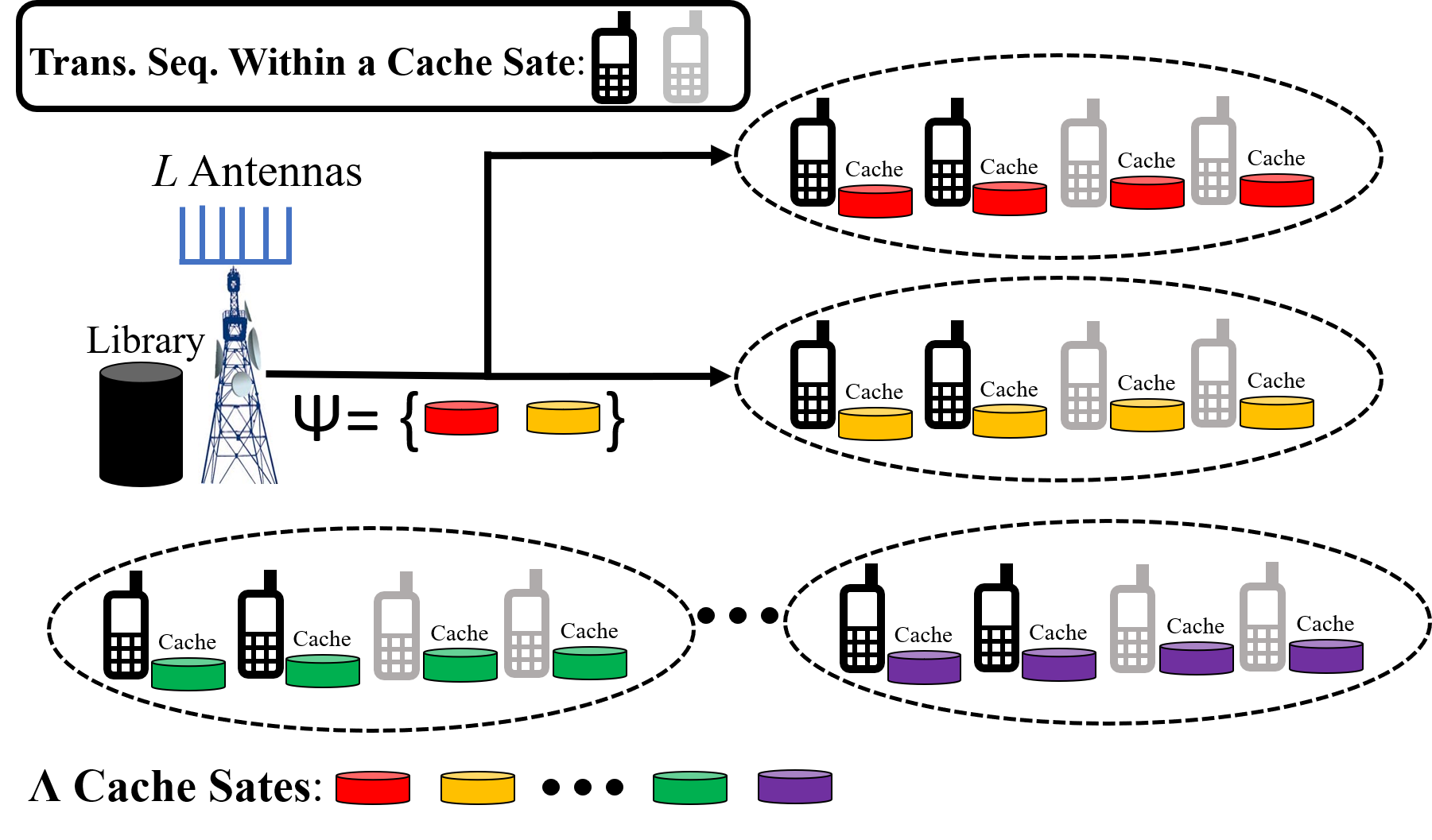}\vspace{-0.2cm}
\captionsetup{font={footnotesize}}
\caption{Illustration for vector coded caching with $G=2$, $B=4$ and $Q=2$.}
\label{Example_VectorCC} 
\end{figure}

\subsection{BD Precoding and MRC Combining}\label{BD_MRC_sub}
As pointed out in \cite{Haardt_BD}, complete channel diagonalization (e.g., ZF) at the BS is suboptimal since each multi-antenna user is able to coordinate the processing of its own receiver outputs.
We thus alternatively consider the well-known BD precoding method~\cite{Haardt_BD}. Toward this, we first define the matrix 
$ {\bf H}_{\psi,-k} \triangleq  [ {\bf H}_{\psi,1}, \cdots,  {\bf H}_{\psi,k-1},  {\bf H}_{\psi,k+1}, \cdots,  {\bf H}_{\psi,Q}  ] \in \mathbb{C}^{ L \times (M_\psi-M_{\psi,k}) }$. To cancel the intra-group interference in \eqref{received_sig_Multi}, ${\bf V}_{\psi,k}$ must lie in the  null-space of ${\bf H}_{\psi,-k}^*$ (denoted by $\mathbb{O}\{ {\bf H}_{\psi,-k}^* \}$) such that the product of ${\bf H}_{\psi,-k}^T$ and ${\bf V}_{\psi,k}$ is the zero matrix for any $k \in [Q]$. Thus, to successfully eliminate the inter-stream interference in ${\rm U}_{\psi,k}$, we have that
\begin{align}\label{J_constraint}
    J_{\psi,k} \le \min  \Big\{ \text{Rank}\big(\mathbb{O}\{ {\bf H}_{\psi,-k}^* \} \big), \ \text{Rank}\big({\bf H}_{\psi,k} \big) \Big\},
\end{align}
which, for independent Rayleigh fading channels, implies that $J_{\psi,k} \le \min \big\{L - (M_\psi - M_{\psi,k}), M_{\psi,k} \big\}$. We use $Q_{\psi,{\rm max}}$ to denote the maximum $Q$ in the user-group $\psi$, which equals $\min\big\{ \lfloor \frac{M+L-1}{M} \rfloor, \ B \big\}$ if all the served users have the same number of receive antennas $M$. As the users are equipped with different numbers of antennas, we determine $Q_{\psi,{\rm max}}$ by solving $\frac{\sum_{k\in [Q]} M_{\psi,k}}{1-L + \sum_{k
\in [Q]} M_{\psi,k}} -Q=0$. As we here consider
the simplified case of a uniform $Q$ for each group, we accept the aforementioned maximum allowable $Q$ that takes the form $Q_{\rm max} = 
 \min_{\psi \in \Psi} Q_{\psi,{\rm max}}$. 

Let ${\bf T}_{\psi,-k} \triangleq {\bf I}_L - {\bf H}_{\psi,-k}^* \big( {\bf H}_{\psi,-k}^T {\bf H}_{\psi,-k}^* \big)^{-1} {\bf H}_{\psi,-k}^T$ be the projection matrix which maps any vector ${\bf m} \in \mathbb{C}^{L}$ into the null-space of ${\bf H}_{\psi,-k}^*$.  We note that ${\bf T}_{\psi,-k}^2 = {\bf T}_{\psi,-k} = {\bf T}_{\psi,-k}^H$ according to the idempotent
and Hermitian properties of a projection matrix. The BD precoding matrix $ {\bf V}_{\psi,k} \in \mathbb{C}^{L \times J_{\psi,k}}$ dedicated to ${\rm U}_{\psi,k}$ can be written as
\begin{align}\label{V_mat_exp}
    {\bf V}_{\psi,k} = \bigg[ \frac{{\bf T}_{\psi,-k} {\bf m}_{\psi,k,1}}{||{\bf T}_{\psi,-k} {\bf m}_{\psi,k,1}||}, \cdots, \frac{{\bf T}_{\psi,-k} {\bf m}_{\psi,k,J_{\psi,k}}}{||{\bf T}_{\psi,-k} {\bf m}_{\psi,k,J_{\psi,k}}||} \bigg]
\end{align}
where the $q$-th $(q \in [J_{\psi,k}])$ column (denoted by ${\bf v}_{\psi,k,q}$) of ${\bf V}_{\psi,k}$ precodes $s_{\psi,k,q}$.
By using this BD precoding, the  signal vector at ${\rm U}_{\psi,k}$ in \eqref{receive_signal_multi} becomes
\begin{align}
     &{\bf y}_{\psi,k}' =   {\bf R}_{\psi,k}^H {\bf H}_{\psi,k}^T {\bf V}_{\psi,k} {\bf P}_{\psi,k} {\bf s}_{\psi,k}   + {\bf z}_{\psi,k}',
\end{align}
and its $q$-th element $y_{\psi,k,q}'$  of ${\bf y}_{\psi,k}'$ (i.e., which corresponds to decoding $s_{\psi,k,q}$) is of the form
\begin{align}\label{de_s_psikq_exp}
   &y_{\psi,k,q}' =  {\bf r}_{\psi,k,q}^H {\bf H}_{\psi,k}^T {\bf v}_{\psi,k,q} \sqrt{P_{\psi,k,q}} s_{\psi,k,q}  +   {z}_{\psi,k,q}' \notag\\
   &\hspace{0.8cm}+ \underbrace{ \sum\nolimits_{p \in [J_{\psi,k}] \setminus q}  {\bf r}_{\psi,k,q}^H {\bf H}_{\psi,k}^T {\bf v}_{\psi,k,p} \sqrt{P_{\psi,k,p}} s_{\psi,k,p} }_{\text{inter-stream interference}} ,
\end{align}
where this inter-stream interference results from the other symbols intended for the same user. To maximize the power of the useful signal for decoding $s_{\psi,k,q}$, we have that $|{\bf r}_{\psi,k,q}^H {\bf H}_{\psi,k}^T {\bf v}_{\psi,k,q}|^2 \le ||{\bf H}_{\psi,k}^T {\bf v}_{\psi,k,q}||^2$, where the equality is achieved only if 
\begin{align}\label{MRC_rec_exp}
{\bf r}_{\psi,k,q} = \theta {\bf H}_{\psi,k}^T {\bf v}_{\psi,k,q} = \theta {\bf H}_{\psi,k}^T \frac{{\bf T}_{\psi,-k} {\bf m}_{\psi,k,q}}{|| {\bf T}_{\psi,-k} {\bf m}_{\psi,k,q} ||}
\end{align}
for a non-zero constant $\theta$ according to the well-known Cauchy–Schwarz inequality.  Here, we set $\theta = 1/||{\bf H}_{\psi,k}^T {\bf v}_{\psi,k,q}||$ to normalize ${\bf r}_{\psi,k,q}$; this corresponds to the MRC receiver. 

In the following, we will show how to select the vectors ${\bf m}_{\psi,k,1},  \cdots, {\bf m}_{\psi,k,J_{\psi,k}}$ in \eqref{V_mat_exp} to remove the inter-stream interference in \eqref{de_s_psikq_exp}.  When ${\bf m}_{\psi,k,1}, \cdots, {\bf m}_{\psi,k,J_{\psi,k}}$ are the \emph{orthogonal} eigenvectors of ${\bf T}_{\psi,-k} {\bf H}_{\psi,k}^* {\bf H}_{\psi,k}^T {\bf T}_{\psi,-k}$, the inter-stream interference in $y_{\psi,k,q}'$ is completely eliminated, which can be easily observed by the fact that for $p\neq q$, we have that
\begin{align}\label{null_r_eq}
    &{\bf r}_{\psi,k,q}^H {\bf H}_{\psi,k}^T {\bf v}_{\psi,k,p} = \! \frac{ \theta {\bf m}_{\psi,k,q}^H {\bf T}_{\psi,-k} {\bf H}_{\psi,k}^* {\bf H}_{\psi,k}^T {\bf T}_{\psi,-k} {\bf m}_{\psi,k,p} } {|| {\bf T}_{\psi,-k} {\bf m}_{\psi,k,q} || \ || {\bf T}_{\psi,-k} {\bf m}_{\psi,k,p} ||} \notag\\
    &\hspace{1.9cm}= \frac{ \theta \lambda_{\psi,k,p} {\bf m}_{\psi,k,q}^H {\bf m}_{\psi,k,p} } {|| {\bf T}_{\psi,-k} {\bf m}_{\psi,k,q} || \ || {\bf T}_{\psi,-k} {\bf m}_{\psi,k,p} ||} =0
\end{align}
where $\lambda_{\psi,k,p}$ denotes the eigenvalue associated with the eigenvector ${\bf m}_{\psi,k,p}$ of ${\bf T}_{\psi,-k} {\bf H}_{\psi,k}^* {\bf H}_{\psi,k}^T {\bf T}_{\psi,-k}$.
Therefore, under the common Gaussian signaling and the coherent BD-MRC scheme, the signal-to-interference plus noise ratio (SINR) for decoding $s_{\psi,k,q}$ at ${\rm U}_{\psi,k}$ is of the form
\begin{align}
    &\text{SINR}_{\psi,k,q}^{\text{BD-MRC}} \notag\\
     &\hspace{0.5cm}= \frac{P_{\psi,k,q} \ \big| {\bf r}_{\psi,k,q}^H {\bf H}_{\psi,k}^T {\bf v}_{\psi,k,q} \big|^2}{N_0 + \sum\nolimits_{p \in [J_{\psi,k}] \setminus q} P_{\psi,k,p} \ \big| {\bf r}_{\psi,k,q}^H {\bf H}_{\psi,k}^T {\bf v}_{\psi,k,p} \big|^2  } \label{SINR_BDMRC_def} \\
    &\hspace{0.5cm}= \frac{P_{\psi,k,q}}{N_0} \frac{ {\bf m}_{\psi,k,q}^H {\bf T}_{\psi,-k} {\bf H}_{\psi,k}^* {\bf H}_{\psi,k}^T {\bf T}_{\psi,-k} {\bf m}_{\psi,k,q}}{||{\bf T}_{\psi,-k} {\bf m}_{\psi,k,q}||^2}\label{SINR_BDMRC_def2}
\end{align}
where we note that \eqref{SINR_BDMRC_def} is based on the received signal in \eqref{de_s_psikq_exp}, while \eqref{SINR_BDMRC_def2} follows from \eqref{MRC_rec_exp} and \eqref{null_r_eq}.

Toward providing CSI estimates to the BS and the users, we will consider the common TDD uplink-downlink pilot transmission, as this applies to MU-MIMO systems. Thus for $T$ being the coherence block period (in symbols), and for $\Theta$ being the number of symbols per receiving antenna and per block used for pilot transmission, then the effective rate for ${\rm U}_{\psi,k}$ is 
\begin{align}
    R_{\psi,k} = \xi_{G,Q} \sum\nolimits_{q =1}^{J_{\psi,k}}  \ln \Big( 1+\text{SINR}_{\psi,k,q}^{\text{BD-MRC}} \Big) \text{ nats/s/Hz}
\end{align}
where $\xi_{G,Q} \triangleq 1- \Theta \big( \sum_{\psi \in \Psi} \sum_{k \in [Q]} M_{\psi,k} \big) /T$ accounts for CSI costs \cite{caire2010multiuser}. 
Without loss of generality, we sort $\{s_{\psi,k,q}: q \in [J_{\psi,k}]\}$  in descending order according to the corresponding  $\{\text{SINR}_{\psi,k,q}^{\text{BD-MRC}}: q \in [J_{\psi,k}]\}$ at ${\rm U}_{\psi,k}$.

\subsection{Some Preliminary Results for BD-MRC}
As we see, implementing BD-MRC requires us to eigen-decompose ${\bf T}_{\psi,-k} {\bf H}_{\psi,k}^* {\bf H}_{\psi,k}^T {\bf T}_{\psi,-k}$ whose large size (comprising $L \times L$ elements) renders it computationally consuming, especially in the massive MIMO regime.   This motivates us to design a low-complexity approach to considerably ease the BD-MRC implementation in Lemma~\ref{BD_MRC_design_lem}.
 
This design is motivated by the observation (which can be shown by using basic properties of projection matrices and basic algebraic manipulations) that the matrices ${\bf H}_{\psi,k}^T {\bf T}_{\psi,-k} {\bf H}_{\psi,k}^*$ and ${\bf T}_{\psi,-k} {\bf H}_{\psi,k}^* {\bf H}_{\psi,k}^T {\bf T}_{\psi,-k}$ share the same non-zero eigenvalues. 
Before presenting Lemma~\ref{BD_MRC_design_lem}, we define 
the vectors $\{ {\bf t}_{\psi,k,q} \in \mathbb{C}^{M_{\psi,k}}:  q \in [J_{\psi,k}]\}$  as the orthogonal eigenvectors of ${\bf H}_{\psi,k}^T {\bf T}_{\psi,-k} {\bf H}_{\psi,k}^*$  $\in \mathbb{C}^{M_{\psi,k} \times M_{\psi,k}}$, where the eigenvector ${\bf t}_{\psi,k,q}$ is associated with the $q$-th largest (non-zero) eigenvalue $\lambda_{\psi,k,q}$.  

\begin{lemma}\label{BD_MRC_design_lem}
The precoding vector ${\bf v}_{\psi,k,q}$  to $s_{\psi,k,q}$ for ${\rm U}_{\psi,k}$ under the BD-MRC scheme is given the form
\begin{align}\label{v_psikq_design_exp}
    {\bf v}_{\psi,k,q} =
    \frac{{\bf T}_{\psi,-k}  {\bf H}_{\psi,k}^* {\bf t}_{\psi,k,q}}{||{\bf T}_{\psi,-k}  {\bf H}_{\psi,k}^* {\bf t}_{\psi,k,q}||},
\end{align}
and this yields a corresponding SINR for symbol $s_{\psi,k,q}$ of the form
\begin{align}\label{SINR_BD_MRC_lem}
   \text{SINR}_{\psi,k,q}^{\text{BD-MRC}} =  \frac{P_{\psi,k,q}}{N_0} \lambda_{\psi,k,q},
\end{align}
and an effective rate for ${\rm U}_{\psi,k}$ of the form
\begin{align}\label{Rate_Multiplex_Single}
    R_{\psi,k} = \xi_{G,Q} \sum\nolimits_{q=1}^{J_{\psi,k}} \ln \left( 1+ \frac{P_{\psi,k,q}}{N_0} \lambda_{\psi,k,q} \right).
\end{align}
\end{lemma}

\begin{proof}
The proof is relegated to Appendix~\ref{Proof_BD_MRC_design_lem}.
\end{proof}

 \begin{remark}\label{power_allo_remark}
We can see that each user ${\rm U}_{\psi,k}$ can remove interference to decode up to $J_{\psi,k} \leq \text{Rank} \big( {\bf H}_{\psi,k}^T {\bf T}_{\psi,-k} {\bf H}_{\psi,k}^* \big) \leq M_{\psi,k}$ of its symbols. We can also see that under the constraint of a total power $P_{\psi,k} = \sum_{q=1}^{J_{\psi,k}} P_{\psi,k,q}$ allocated to ${\rm U}_{\psi,k}$, water-filling (cf. \cite[Ch. 10]{Goldsmith}) can maximize the effective rate of that user.
 \end{remark}

\begin{remark}\label{BD_MRC_Order}
We note that the design in Lemma~\ref{BD_MRC_design_lem} involves an eigendecomposition of a much smaller $M_{\psi,k}\!\times\! M_{\psi,k}$ matrix with a computational complexity of $\mathcal{O}(M_{\psi,k}^3)$. Assuming each user has $M$ receive antennas and $J_{\psi,k} = M$, the complexity of \eqref{v_psikq_design_exp} is $\mathcal{O}(M^3 + L^2 + ML)$, while that of ${\bf V}_{\psi,k}$ is $\mathcal{O}(M^3 + ML^2 + M^2L)$. In contrast, eigendecomposing ${\bf T}_{\psi,-k}{\bf H}_{\psi,k}^*{\bf H}_{\psi,k}^T{\bf T}_{\psi,-k}$ requires $\mathcal{O}(L^3)$, and the cost of \eqref{V_mat_exp} is $\mathcal{O}(L^3 + ML^2)$. Since $L \gg M$ in practice, the overall scaling is asymptotically $\mathcal{O}(L^2)$ for Lemma~\ref{BD_MRC_design_lem} and $\mathcal{O}(L^3)$ for \eqref{V_mat_exp}, yielding an order-$L$ reduction in computational complexity.
\end{remark}


Let us now address the balance between multiplexing and beamforming gains, by first noting that such optimal balancing requires either real-time adjustment of the number of transmitted symbols ($J_{\psi,k}$), or, when we always use the maximum allowable number of transmitted symbols for ${\rm U}_{\psi,k}$, requires water-filling to dynamically allocate different powers to our symbols.
 We will here consider the maximum allowable value of $J_{\psi,k}$ transmitted symbols to ${\rm U}_{\psi,k}$, i.e., we will consider
 \begin{align}
     J_{\psi,k} = \text{Rank} \big( {\bf H}_{\psi,k}^T {\bf T}_{\psi,-k} {\bf H}_{\psi,k}^* \big) \ \le M_{\psi,k}.
 \end{align}


The following corollary characterizes the optimal power allocation for \( {\rm U}_{\psi,k} \) by leveraging the water-filling algorithm, which distributes power across the transmitted symbols based on instantaneous channel conditions, ensuring adherence to the power constraint \( P_{\psi,k} \) and maximizing the effective rate in~\eqref{Rate_Multiplex_Single}.

\begin{coro}\label{water_filling_coro}
The optimal power allocated to symbol $s_{\psi,k,q}$ for maximizing the effective rate in \eqref{Rate_Multiplex_Single}, takes the form
\begin{align}\label{water_filling_single}
    P_{\psi,k,q} = \Big( \frac{1}{\alpha_{\psi,k}} - \frac{N_0}{\lambda_{\psi,k,q}} \Big)^+,
\end{align}
where, under a power constraint $P_{\psi,k}$, the Lagrange multiplier $\alpha_{\psi,k}$ is the solution to 
\begin{align}
    \sum\nolimits_{q=1}^{J_{\psi,k}}  \Big( \frac{1}{\alpha_{\psi,k}} - \frac{N_0}{\lambda_{\psi,k,q}} \Big)^+ = P_{\psi,k}.
\end{align}
Then, the optimal effective rate for ${\rm U}_{\psi,k}$ takes the form
\begin{align}\label{R_single_water_filling}
    R_{\psi,k}^\star (P_{\psi,k}) = \xi_{G,Q}\sum\limits_{q=1}^{J_{\psi,k}} \ln \left( 1 + \Big( \frac{\lambda_{\psi,k,q}}{N_0 \alpha_{\psi,k}} -1 \Big)^+ \right).
\end{align}
\end{coro}
 
\begin{proof}
The proof is direct from the water-filling algorithm.
\end{proof}

\subsection{MMF Framework and Main Performance Metrics}
Let us now further consider MMF, where the minimum effective rate among the simultaneously served users is maximized via power allocation under a specific precoding scheme. The MMF problem for serving the users $\{{\rm U}_{\psi,k}: \psi \in \Psi, k \in [Q]\}$ can be now formulated as
\begin{align}\label{MMF_model}
\mathcal{S}_1 \begin{cases}
     &\max\nolimits_{\mathcal{P}_\Psi} \min\nolimits_{\psi \in \Psi} \min_{k \in [Q]} \ \ R_{\psi,k}  \\
    &\text{s.t. } P_t =  \sum\nolimits_{\psi \in \Psi} \sum\nolimits_{k \in [Q]} P_{\psi,k} \le P_{\rm tot}
\end{cases}
\end{align}
where $P_{\rm tot}$ denotes the maximum allowable transmit power at the BS, and $\mathcal{P}_\Psi \triangleq \{ P_{\psi,k}: \ \psi \in \Psi, k \in [Q]   \}$ denotes the set of used powers.

\begin{remark}
The water-filling algorithm in Corollary~\ref{water_filling_coro} is applied only across the streams of a given user ${\rm U}_{\psi,k}$ to maximize that user's effective rate under its allocated power $P_{\psi,k}$, i.e., the total power used to transmit the signal streams to ${\rm U}_{\psi,k}$. User fairness is ensured separately by the MMF formulation in \eqref{MMF_model}, which maximizes the minimum effective rate among the concurrently served $GQ$ users in VCC.
\end{remark}


Let us now formally define some important metrics of interest, which will be analyzed in Sections~\ref{multi_ana_sec} and \ref{single_ana_sec}.

\begin{definition}\label{sum_rate_def}
\emph{(Effective sum-rate).} For a $(G,Q)$-VCC scheme, its effective (instantaneous) sum-rate is denoted by $R(G,Q)$ and is defined as the total effective rate (after accounting for CSI costs) summed over the $GQ$ simultaneously served users. Moreover, $ R^\star(G,Q)$ denotes the effective sum-rate optimized under the MMF criterion (cf.~\eqref{MMF_model}).  
\end{definition}
\begin{definition}\label{effective_def}
\emph{(Effective gain over MU-MIMO).} For a given set of SNR, $L$ and $M_{\psi,k} \ (\forall \psi \in \Psi, k \in [Q])$  resources, and a fixed underlying precoder class, the effective gain, after accounting for CSI costs, of the $(G,Q)$-VCC scheme over the cacheless scenario (corresponding to $G = 1$, and an operating multiplexing gain $Q'$), will be denoted as $\Gc \triangleq \frac{\mathbb{E}_{h,r}\{ { R}^\star(G,Q) \}}{ \mathbb{E}_{h,r}\{ {R}^\star(1,Q')\} }$, where $\mathbb{E}_{h,r}\{ R^\star(G,Q) \}$ describes the rate  $R^\star(G,Q)$ averaged over channel fading and pathloss. We also call $\Gc^\star \triangleq \frac{\max_Q  \mathbb{E}_{h,r}\{ { R}^\star(G,Q)\} }{\max_{Q'} \mathbb{E}_{h,r}\{ {R}^\star(1,Q')\} }$ as the effective gain of optimized rates, where $Q'$ and $Q$ are also optimized, independently for the cacheless and the cache-aided scenario.
\end{definition}

\section{Performance Analysis on BD-MRC Based Vector Coded Caching}\label{multi_ana_sec}
This section begins by analyzing and solving the MMF problem under the BD-MRC based VCC framework, which is reformulated into a one-dimensional linear search problem. In addition, we derive analytical upper and lower bounds for the optimal MMF solution. We then examine two special scenarios: (i) the case where each served user is sent the same number of symbols in each transmission round, and (ii) the massive MIMO scenario with Rayleigh fading channels.

\subsection{Main Results}
Let us first recall (cf.~\eqref{null_r_eq}) that the eigenvalue set $\{ \lambda_{\psi,k,q}: {\psi \in \Psi, k \in [Q]}, q \in [J_{\psi,k}] \}$ is a function of the channel gains but not of the power allocation policy.
Thus, by using the  effective rate expression in \eqref{Rate_Multiplex_Single}, the MMF optimization problem in \eqref{MMF_model} under the BD-MRC scheme for downlink power allocation, can be transformed into
\begin{align}\label{BD_MRC_Multiplex_Opt}
\mathcal{S}_2  \begin{cases}
    &\max\limits_{\mathcal{P}_\Psi} \min\limits_{\psi \in \Psi} \min\limits_{k \in [Q]} \xi_{G,Q} \sum\limits_{q=1}^{J_{\psi,k}} \ln \left( 1+ \frac{P_{\psi,k,q}}{N_0} \lambda_{\psi,k,q} \right) \\
    & \text{s.t. } P_t = \sum\limits_{\psi \in \Psi} \sum\limits_{k \in [Q]} \sum\limits_{q \in [J_{\psi,k}]} P_{\psi,k,q} \le P_{\rm tot}.
\end{cases}
\end{align}
Obviously, $P_t$  should reach its upper-bound  $P_{\rm tot}$  when the optimum is achieved.
The following theorem solves the optimization problem in~\eqref{BD_MRC_Multiplex_Opt} to derive the optimal effective (instantaneous) sum-rate $R^\star_{\text{BD-MRC}}$ and the effective gain of optimized rates (cf. Definition~\ref{effective_def}). Before presenting the theorem, let us note that $f_{\psi,k}^{-1}(\cdot)$ will denote the inverse function of $R^\star_{\psi,k}(P_{\psi,k})$ in~\eqref{R_single_water_filling}, which is a monotonically increasing function w.r.t. $P_{\psi,k}$. We proceed with the theorem.

\begin{theorem}\label{R_BD_MRC_opt_Thm}
The effective sum-rate $R^\star_{\text{BD-MRC}}$ under optimal power allocation for the MMF problem  is the solution to 
\begin{align}\label{R_BD_MRC_num_eq}
    \sum\nolimits_{\psi \in \Psi} \sum\nolimits_{k \in [Q]} f_{\psi,k}^{-1}\Big(\frac{R_{\text{BD-MRC}}^\star}{GQ}\Big) = P_{\rm tot},
\end{align}
and the optimal effective rate for each simultaneously served user is identical, given by ${R_{\text{BD-MRC}}^\star}/(GQ)$.
The corresponding  effective gain of optimized rates under the BD-MRC scheme is given by
\begin{align}
    \mathcal{G}_{\text{BD-MRC}}^\star = \frac{\max_{Q \in [Q_{\max}]} \mathbb{E}_{h,r}\{R^\star_{\text{BD-MRC}}(G,Q)\}}{\max_{Q' \in [Q_{\max}']} \mathbb{E}_{h,r}\{R^\star_{\text{BD-MRC}}(1,Q')\}}.
\end{align}
Furthermore, this optimal sum-rate is bounded as 
\[\widetilde{R}^\star_{\text{BD-MRC}} \le R^\star_{\text{BD-MRC}} \le \widehat{R}^\star_{\text{BD-MRC}}\] 
where $\widetilde{R}^\star_{\text{BD-MRC}}$ and $\widehat{R}^\star_{\text{BD-MRC}}$ are respectively the solutions to 
\begin{align}
    &\sum_{\psi \in \Psi} \sum_{k \in [Q]} \frac{J_{\psi,k} N_0}{\lambda_{\psi,k}^{\min}} \left(\exp\left( \frac{\widetilde{R}^\star_{\text{BD-MRC}}}{\xi_{G,Q} J_{\psi,k} G Q} \right) -1 \right) = P_{\rm tot}, \\ 
    & \sum_{\psi \in \Psi} \sum_{k \in [Q]} \frac{J_{\psi,k} N_0}{\lambda_{\psi,k}^{\max} } \left(\exp\left( \frac{\widehat{R}^\star_{\text{BD-MRC}}}{\xi_{G,Q} J_{\psi,k} G Q} \right) -1 \right) = P_{\rm tot}, 
\end{align}
where $\lambda_{\psi,k}^{\min} \triangleq \min\limits_{q \in [J_{\psi,k}]}\{\lambda_{\psi,k,q}\}$ and $\lambda_{\psi,k}^{\max} \triangleq \max\limits_{q \in [J_{\psi,k}]}\{\lambda_{\psi,k,q}\}$.
\end{theorem}
\begin{proof}
As the power allocation among the symbols intended by ${\rm U}_{\psi,k}$ does not affect the power allocation to other served users, the effective rate for ${\rm U}_{\psi,k}$ must reach its optimal bound under the power constraint $P_{\psi,k}^\star$ (optimal power allocated to ${\rm U}_{\psi,k}$ for MMF), which has been solved in Corollary \ref{water_filling_coro}.
As $R_{\psi,k}^\star (P_{\psi,k})$ in \eqref{R_single_water_filling} is a monotonically increasing function w.r.t. $P_{\psi,k}$, we can conclude that the simultaneously served users must have the same effective rate (equaling $R_{\text{BD-MRC}}^\star/G/Q$) under the optimal power allocation in \eqref{BD_MRC_Multiplex_Opt}. To see this, simply consider a contradictory case where, when the optimum for  \eqref{BD_MRC_Multiplex_Opt} is achieved,
if there exists one user whose effective rate is higher than the smallest rate, this user can ``borrow" some power to the user with the smallest rate until their rates are the same, without affecting the rates for other users, which enhances the smallest effective rates, and which is contradictory to the optimal power allocation assumption.   Considering $\sum_{\psi\in \Psi} \sum_{k \in [Q]} P_{\psi,k}^\star = P_{\rm tot}$ and the inverse function of $R^\star_{\psi,k}(P_{\psi,k})$, we can obtain \eqref{R_BD_MRC_num_eq}. 

To derive the lower-bound for $R_{\psi,k}$, considering \eqref{Rate_Multiplex_Single}, we first have that 
\begin{align}
  R_{\psi,k} 
  &\ge   \xi_{G,Q} \sum\nolimits_{q=1}^{J_{\psi,k}} \ln \left( 1+ \frac{P_{\psi,k,q}}{N_0} \lambda_{\psi,k}^{\min} \right)
\end{align}
which is then substituted into~\eqref{BD_MRC_Multiplex_Opt} as the objective function. Using similar analysis as that leading to \eqref{R_BD_MRC_num_eq}, we can easily obtain the identity for $\widetilde{R}^\star_{\text{BD-MRC}}$. 
The upper-bound $\widehat{R}^\star_{\text{BD-MRC}}$ follows the same procedure as $\widetilde{R}^\star_{\text{BD-MRC}}$, with a difference being that upper bounding $R_{\psi,k}$ now uses $ \lambda_{\psi,k}^{\max}$.
\end{proof}

As shown in Fig.~\ref{Sum_Rate_fig}, the numerical results under different system parameter settings exhibit an excellent match with the simulation results, which confirms the accuracy of Theorem~\ref{R_BD_MRC_opt_Thm}.

\begin{remark}
We quickly ask the reader to note that the derived sum-rate bounds $\widetilde{R}^\star_{\text{BD-MRC}}$ and $\widehat{R}^\star_{\text{BD-MRC}}$ considerably facilitate the use of the simple one-dimensional binary search to numerically establish $R^\star_{\text{BD-MRC}}$ in \eqref{R_BD_MRC_num_eq}.    
\end{remark}

\subsection{Special Case I: Symmetric Case }
We also have the following, which considers the commonly-assumed symmetric case where each scheduled user is sent an equal number of symbols.

\begin{coro}\label{upper_lower_sim_coro_BD_MRC}
In the symmetric case where $J_{\psi,k}= J$ for any $\psi \in \Psi$ and $k \in [Q]$, the lower and upper bounds to the optimal sum-rate $R^\star_{\text{BD-MRC}}$ take the forms
\begin{align}
    &\widetilde{R}^\star_{\text{BD-MRC}} \!\triangleq\!  \xi_{G,Q} G Q J \ln\!\Bigg( \!  1 \!+\! \frac{P_{\rm tot} /(N_0 J) }{ \sum\limits_{\psi \in \Psi} \sum\limits_{k \in [Q]}   (\lambda_{\psi,k}^{\min})^{-1} } \! \Bigg), \\
    &\widehat{R}^\star_{\text{BD-MRC}} \!\triangleq\!  \xi_{G,Q} G Q J \ln\!\Bigg(\! 1 \!+\! \frac{P_{\rm tot} /(N_0 J)}{ \sum\limits_{\psi \in \Psi} \! \sum\limits_{k \in [Q]} (\lambda_{\psi,k}^{\max})^{-1} } \! \Bigg).
\end{align}
\end{coro}

\begin{proof}
The derivation of $\widetilde{R}^*_{\text{BD-MRC}}$ and  $ \widehat{R}^*_{\text{BD-MRC}}$ is direct after setting $J_{\psi,k} = J, \ \forall \psi \in \Psi, k \in [Q]$ in Theorem~\ref{R_BD_MRC_opt_Thm}. 
\end{proof}

\begin{remark}
We note that Lemma~\ref{BD_MRC_design_lem}, Corollary~\ref{water_filling_coro}, Theorem~\ref{R_BD_MRC_opt_Thm} and Corollary~\ref{upper_lower_sim_coro_BD_MRC}, are all valid for any propagation channel model, including Rayleigh fading, Rician-$K$ fading, etc. The same results also hold for scenarios that involve a non-full-rank channel matrix product ${\bf H}_{\psi,-k}^T {\bf H}_{\psi,-k}^*$ (e.g.,  Keyhole channels) for $\psi \in \Psi$ and $k \in [Q]$, in which case we apply the pseudo-inverse of ${\bf H}_{\psi,-k}^T {\bf H}_{\psi,-k}^*$ in the projection matrix ${\bf T}_{\psi,-k}$.
\end{remark}

\subsection{Special Case II: Massive MIMO Over Rayleigh Channels}
In the following, we consider the massive MIMO regime~\cite{Rusek} entailing a very large  $L$, generally corresponding to the case of $L \gg M_\psi = \sum_{k \in [Q]} M_{\psi,k}$ for $\forall \psi \in \Psi$. This setting captures certain technological trends, as well as simplifies exposition, by simplifying analysis of the eigenvalues of ${\bf H}_{\psi,k}^T {\bf T}_{\psi,-k} {\bf H}_{\psi,k}^* \in \mathbb{C}^{M_{\psi,k} \times M_{\psi,k}}$. We here consider independent Rayleigh fading channels where $J_{\psi,k}=M_{\psi,k}$ and where the elements of ${\bf H}_{\psi,k} $ follow the i.i.d. complex Gaussian distribution with zero-mean and variance $\beta_{\psi,k}$, where $\beta_{\psi,k}$ accounts for the large-scale fading and/or pathloss. Lemma \ref{BD_Lem1} distills the result of Theorem~\ref{R_BD_MRC_opt_Thm} to the massive MIMO case, and the reported results are naturally independent of instantaneous channel fading.

\begin{lemma}\label{BD_Lem1}
In the massive MIMO regime, the effective instantaneous sum-rate $R_{\text{BD-MRC}}^\star$  under the BD-MRC scheme and optimal power allocation for MMF over independent Rayleigh fading channels, is the solution to 
\begin{align}\label{R_BD_MRC_LargeL_eq}
    \sum\limits_{\psi \in \Psi} \sum\limits_{k \in [Q]}   \frac{N_0 M_{\psi,k} \left( \exp\left( \frac{R_{\text{BD-MRC}}^\star}  {\xi_{G,Q} M_{\psi,k} G Q} \right) -1 \right)   }{\beta_{\psi,k} \big( L - M_\psi + M_{\psi,k}  \big)} = P_{\rm tot}
\end{align}
with the corresponding power allocation policy that yields the asymptotically optimal $R_{\text{BD-MRC}}^\star$, taking the form
\begin{align}\label{Power_allo_BD_LargeL_eq}
    P_{\psi,k,q} =  \frac{N_0 M_{\psi,k} \left( \exp\left( \frac{R_{\text{BD-MRC}}^\star} { \xi_{G,Q} M_{\psi,k}  G Q} \right) -1 \right)   }{\beta_{\psi,k} \big( L - M_\psi + M_{\psi,k}  \big)} . 
\end{align}
For the case when $M_{\psi,k} = M$ for $\forall \psi \in \Psi, k \in [Q]$, then \begin{align}\label{R_BD_MRC_LargeL_M_eq}
    &R_{\text{BD-MRC}}^\star (G,Q) \notag\\
    &\simeq  \xi_{G,Q} G Q  M \ln\Bigg( 1+ \frac{P_{\rm tot}(L-(Q-1)M)}{N_0 M \sum_{\psi \in \Psi} \sum_{k \in [Q]} \beta_{\psi,k}^{-1}} \Bigg).
\end{align}
\end{lemma}

\begin{proof}
The proof is relegated to Appendix~\ref{Proof_Lem1}.
\end{proof}

The numerical results shown in Fig.~\ref{Gain_L128_1}, under various system parameter settings, closely match the simulation results, confirming the accuracy of Lemma~\ref{BD_Lem1}.

\section{Performance Analysis on simple ZF Based Vector Coded Caching}\label{single_ana_sec}

In this section, we analyze the effective sum-rate achieved by the cache-aided downlink schemes of Section~\ref{JSAC_Intro} for the case of the ZF linear precoder over independent Rayleigh fading channels where $J_{\psi,k} = M_{\psi,k}$ for $\psi \in \Psi$ and $k \in [Q]$.  In contrast to the coherent BD-MRC design where both the precoding/combining matrix optimization and power allocation should be adjusted according to the instantaneous channel state (recall Theorem~\ref{R_BD_MRC_opt_Thm}),  we here (cf.~Theorem~\ref{Rate_multiZF_Thm}) simply perform channel matrix inversion without any precoder optimization, and we simply calibrate the power allocation as a function of the (large-scale) pathloss which of course changes much slower than fading does. The numerical results in Section~\ref{numerical_sec} show that the performance gap between the simpler ZF precoder and the fully optimized BD-MRC scheme is negligible for the practical setting of $L \gg M_\psi$ (cf. Fig.~\ref{Sum_Rate_fig}).

Following the conventional ZF precoding for single-antenna receivers (cf.~\cite{Caire_ZF}), we completely separate the transmitted $M_\psi = \sum_{k \in [Q]} M_{\psi,k}$ symbol streams such that there is no inter-stream interference. Therefore, the $M_{\psi,k}$ symbols simultaneously sent to ${\rm U}_{\psi,k}$ are fully separated (using complete channel diagonalization at the BS), and user ${\rm U}_{\psi,k}$ independently decodes the intended $M_{\psi,k}$ symbols without inter-stream interference. This corresponds to a decoding matrix ${\bf R}_{\psi,k} = {\bf I}_{M_{\psi,k}}$ at user ${\rm U}_{\psi,k}$.      
We use ${\bf H}_{\psi} \triangleq \big[ {\bf H}_{\psi,1}, \cdots, {\bf H}_{\psi,Q}  \big] \in \mathbb{C}^{L \times  M_{\psi}}$ to represent the channel matrix between the BS and the $Q$ active users in user-group $\psi$. Here, the precoding matrix for the multi-antenna receivers in user-group $\psi$ in \eqref{signal_transmit_multi} under the ZF variant scheme is designed as
\begin{align}\label{ZF_Design_Multi}
    {\bf V}_\psi =
     {{\bf H}}^*_{\psi} \left( {{\bf H}_{\psi}^T} {{\bf H}_{\psi}^*} \right)^{-1}  \circ {\bf D}_{\psi}  \     \in \mathbb{C}^{L \times  M_{\psi}}
\end{align}
where $\circ$ denotes the Hadamard product (element-wise product). In \eqref{ZF_Design_Multi},
${\bf D}_{\psi} \in \mathbb{C}^{L \times  M_{\psi}}$ is the normalization matrix which guarantees a unit norm for each column of ${\bf V}_\psi$, and whose $\ell$-th ($\ell \in [M_\psi]$) column  takes the form
$
     {\bf D}_{\psi,\ell}  \triangleq  {\bf 1}_L \big(\big[ \big( {\bf H}_{\psi}^T {\bf H}_{\psi}^* \big)^{-1} \big]_{\ell,\ell} \big)^{-1/2},
$
where ${\bf 1}_L \in \mathbb{C}^L$ is the vector with all elements equaling 1.

\begin{remark}
Computing the ZF precoder ${\bf V}_\psi$ in \eqref{ZF_Design_Multi} has complexity $\mathcal{O}(M_\psi^3 + L M_\psi^2)$. Under the common massive-MIMO regime with $L \gg M_\psi$, the complexity order is asymptotically $\mathcal{O}(L)$. In contrast, BD-MRC for the same group $\psi$ scales asymptotically as $\mathcal{O}(L^2)$, as discussed in Remark~\ref{BD_MRC_Order}. Thus, the ZF design asymptotically achieves an order-$L$ reduction in computational complexity compared to performing BD-MRC via Lemma~\ref{BD_MRC_design_lem}.
\end{remark}

We first present the upper and lower bounds for the effective rate at any given user, averaged over channel fading. 
\begin{prop}\label{Prop_ZF_singe_rate_multi}
The effective rate $\bar R_{\psi,k}^\text{ZF}$ averaged over channel fading at user ${\rm U}_{\psi,k}$ for $\psi \in \Psi, k \in [Q]$ under ZF-based precoding, is bounded as $\widetilde R_{\psi,k}^\text{ZF}  \le \bar R_{\psi,k}^\text{ZF} \le \widehat R_{\psi,k}^\text{ZF}$, where 
\begin{align}
    &\widetilde R_{\psi,k}^\text{ZF} \triangleq \xi_{G,Q} \sum_{q=1}^{M_{\psi,k}} \ln \left( 1+  \frac{ P_{\psi,k,q} (L-M_\psi)   \beta_{\psi,k}} {N_0 } \right), \label{lower_ZF_multi}\\
    &\widehat R_{\psi,k}^\text{ZF} \triangleq \xi_{G,Q} \! \sum_{q=1}^{M_{\psi,k}} \ln \!\left( 1+  \frac{ P_{\psi,k,q} (L-M_\psi+1)   \beta_{\psi,k}} {N_0  } \right). \label{upper_ZF_multi}
\end{align}
\end{prop}
\begin{proof}
The proof is relegated to Appendix~\ref{Proof_Prop_ZF_single_rate}.
\end{proof}


Since directly solving the MMF optimization problem in \eqref{MMF_model} is computationally prohibitive, especially in large-scale settings with many served users and large antenna arrays, we take an alternative approach by optimizing strict upper and lower bounds on the effective sum-rate. Specifically, instead of solving the original MMF problem directly, we formulate two separate optimization problems by treating the upper and lower bounds derived in Proposition~\ref{Prop_ZF_singe_rate_multi} as the objective functions within the MMF framework in Theorem~\ref{Rate_multiZF_Thm}. This allows us to obtain analytically optimized bounds on the sum-rate after MMF optimization, thereby providing rigorous performance guarantees. Furthermore, when all users are equipped with the same number of receive antennas, the optimized bounds lead to closed-form expressions, significantly simplifying performance analysis.
We now proceed with the theorem.

\begin{theorem}\label{Rate_multiZF_Thm}
The optimal MMF-constrained effective sum-rate $\bar R_{\text{ZF}}^\star$ is bounded as $\widetilde{R}_{\text{ZF}}^\star\leq  \bar R_{\text{ZF}}^\star \leq    \widehat{R}_{\text{ZF}}^\star$, where $\widetilde{R}_{\text{ZF}}^\star$ and $\widehat{R}_{\text{ZF}}^\star$ are respectively the solutions to\footnote{{\color{black}{Unlike in Theorem~\ref{R_BD_MRC_opt_Thm} for BD-MRC, where the MMF optimization is performed before averaging over channel fading, we change the order of operations in \( \bar{R}_{\text{ZF}}^\star \). Specifically, we first take the expectation over channel fading and then perform the MMF optimization. This reformulation allows us to design the power allocation strategy based on large-scale fading and path loss. While this approach simplifies the optimization process, it introduces an approximation compared to directly solving the MMF problem before averaging over channel fading, as done in the BD-MRC case.}}} 
\begin{align}
    &\sum_{\psi \in \Psi} \sum_{k \in [Q]} \frac{N_0 M_{\psi,k} \Big( \exp\Big( \frac{\widetilde{R}_{\text{ZF}}^\star}{\xi_{G,Q}GQ M_{\psi,k}}\Big) -1\Big) }{\beta_{\psi,k} (L-M_\psi)}  = P_{\rm tot}, \label{lower_ZF_num}\\
    &\sum_{\psi \in \Psi} \sum_{k \in [Q]} \frac{N_0 M_{\psi,k} \Big( \exp\Big( \frac{\widehat{R}_{\text{ZF}}^\star}{\xi_{G,Q}GQ M_{\psi,k}}\Big) -1\Big) }{\beta_{\psi,k} (L-M_\psi+1)}  = P_{\rm tot}.\label{upper_ZF_num}
\end{align}
In the symmetric case of $M_{\psi,k}= M$ for $ \forall \psi \in \Psi, k \in [Q]$, the closed-form expressions for $\widetilde{R}_{\text{ZF}}^\star$ and $\widehat{R}_{\text{ZF}}^\star$ are respectively
\begin{align}
    &\widetilde{R}_{\text{ZF}}^\star = \xi_{G,Q} GQM \ln\left( 1+ \frac{P_{\rm tot} (L-QM)}{M N_0 \sum_{\psi \in \Psi} \sum_{k \in [Q]} \beta_{\psi,k}^{-1}}\right), \label{lower_ZF_closed} \\
    &\widehat{R}_{\text{ZF}}^\star = \xi_{G,Q} GQM \ln\left( 1+ \frac{P_{\rm tot} (L-QM+1)}{M N_0 \sum_{\psi \in \Psi} \sum_{k \in [Q]} \beta_{\psi,k}^{-1}}\right). 
\end{align}
\end{theorem}

\begin{proof}
The proof follows a similar procedure to the BD-MRC case. We first show that, at the optimum of the corresponding MMF problem, all users attain the same lower bound on their effective rate. Then, by analyzing the lower bound, we derive the transmit power allocated to each user and sum over all served users to match the total power \(P_{\rm tot}\). This leads to~\eqref{lower_ZF_num}. In the symmetric case where all users have the same number of receive antennas \(M\), a closed-form expression for the lower bound in~\eqref{lower_ZF_closed} can be further obtained. The upper bound can be derived using similar steps.
\end{proof}

As shown in Figs. \ref{Sum_Rate_fig},  \ref{Opt_Gain_1_fig}, and \ref{VectorMulti_Micro_fig}, the tightness of the derived bounds in Theorem~\ref{Rate_multiZF_Thm} in estimating the actual performance is validated under various system configurations.

It is important to note that while ZF precoding is particularly beneficial in the large-antenna regime, the results presented in Theorem~\ref{Rate_multiZF_Thm} are not limited to large \( L \). The derived upper and lower bounds remain valid for a wide range of antenna configurations, making them applicable to both moderate and large-scale systems.  These results offer an efficient yet mathematically rigorous way to evaluate MMF-optimized sum-rate behavior while reducing computational complexity.

By comparing \(\widetilde{R}_{\text{ZF}}^\star\) and \(\widehat{R}_{\text{ZF}}^\star\) with \(R_{\text{BD-MRC}}^\star\) (cf.~\eqref{R_BD_MRC_LargeL_M_eq}) in the symmetric setting, we obtain the following limit results:
\begin{align}
&\lim_{L \to \infty}  \frac{| R_{\text{BD-MRC}}^\star -  \widetilde{R}_{\text{ZF}}^\star|}{\xi_{G,Q} GQM}   
= \lim_{L \to \infty} \frac{M}{L - QM} = 0, \\
&\lim_{L \to \infty} \frac{|R_{\text{BD-MRC}}^\star - \widehat{R}_{\text{ZF}}^\star |}{\xi_{G,Q} GQM} = \lim_{L \to \infty} \frac{M-1}{L - QM + 1} = 0.
\end{align}
These results lead to the following remark.

\begin{remark}\label{Small_Gap_ZFBD_remark}
  In the massive MIMO regime (i.e., \(L \to \infty\)), the delivery performance gap between simple ZF and BD-MRC is negligible when the number of receive antennas per user \(M\) is small. This conclusion is further supported by the numerical results in Section~\ref{numerical_sec} (e.g., Figs.~\ref{Sum_Rate_fig} and \ref{Opt_Gain_1_fig}).
\end{remark}

\section{Numerical Results}\label{numerical_sec}
This section presents various numerical results that validate our analysis as well as provide clear comparisons\footnote{The bounds in Theorem~\ref{Rate_multiZF_Thm} have been demonstrated to be very tight in our conference version \cite{Zhao_WSA_SCC}. For clarity of presentation, we only plot the tight lower and upper bounds for ZF precoding as benchmarks to BD-MRC, while omitting the simulation results since they almost perfectly overlap with the bounds and are thus indistinguishable.}. We here focus on the case where each user has the same number of receiving antennas $M$.  We consider users with relatively low mobility, assuming a coherence block of \( T = 15{,}000 \) symbols, which corresponds to a coherence time of \( 0.075 \) seconds and a coherence bandwidth of \( 200 \) kHz.\footnote{For practical reference, at a carrier frequency of $f_c = 2~\text{GHz}$ in an urban Micro-cell scenario, the RMS delay spread (DS) is approximately $114~\text{nanoseconds}$ \cite[Table 7.5-6]{three_GPP}. Using the empirical relation $B_c = 1/(50\,\text{DS})$ for a correlation coefficient of $0.9$ \cite[Eq. (1.23)]{Cho2010MIMOOFDM}, the coherence bandwidth $B_c$ is about $180~\text{kHz}$. Considering a pedestrian moving at speed $v = 3~\text{km/h}$, the maximum Doppler shift is $f_D = v f_c / c \approx 5.5~\text{Hz}$ \cite[Eq. (1.34)]{Cho2010MIMOOFDM}, where $c$ denotes the light speed. This yields a coherence time  of  $T_c = {0.423}/{f_D} \approx 0.077~\text{seconds}$ \cite[Eq. (1.31)]{Cho2010MIMOOFDM}. 
} We further consider CSI pilot length of $\Theta = 10$, which, under some ideal conditions, could provide near-perfect CSI at both the BS and the users~\cite{caire2010multiuser}\footnote{$\Theta$ could be further decreased at higher SNR, thus further reducing the CSI overhead, resulting in even higher performance gains brought about by VCC.}.  We consider independent Rayleigh fading channels and generate 1000 realizations of users' locations, based on the assumption of uniformly-distributed users across the cell.\footnote{To the best of our knowledge, there are currently no existing multi-antenna coded caching schemes that significantly outperform conventional cacheless MU-MIMO in terms of spectral efficiency. Although this work mainly focuses on deriving analytical expressions that parametrize the delivery performance of VCC, we still provide numerical comparisons between VCC and existing multi-antenna coded caching schemes in Section~\ref{Com_BCC_Subsec}.} 
Other simulation parameters are listed in Table~\ref{tab:simulation}.


\begin{table}[!t]
  \caption{Simulation Parameters (cf. \cite{three_GPP,Emil_PathlossModel})}
  \label{tab:simulation}
  \centering
  \footnotesize
  \setlength{\tabcolsep}{4pt}      
  \renewcommand{\arraystretch}{1.15}
  \begin{tabularx}{0.98\linewidth}{@{}lL@{}} 
    \toprule
    \textbf{Parameters} & \textbf{Value} \\
    \midrule
    AWGN spectral density & $-174 \text{ dBm/Hz}$ \\
    Spectrum bandwidth for each user & 20 MHz \\
    Carrier frequency & 2 GHz\\
    Macro-cell size & Inner radius of $35$ m and outer radius of $500$ m\\
    Micro-cell size & Inner radius of $10$ m and outer radius of $100$ m \\
    Pathloss exponent $\eta_0$ & $\eta_0=3.76$ in Macro-cell and  $\eta_0=3$ in Micro-cell \\
    Attenuation regularization $l_0$ & $l_0 = 10^{-3.53}$ in  Macro-cell  and $l_0=10^{-3.7}$ in  Micro-cell \\
    Pathloss model $\beta_{\psi,k}$ &   $\beta_{\psi,k} = l_0 r_{\psi,k}^{-\eta_0}$ where $r_{\psi,k}$ is the distance from the BS to user ${\rm U}_{\psi,k}$\\
    Maximum transmit power $P_{\rm tot}$ & Typical values are around 33 dBm in Micro-cell, and 40 dBm in Macro-cell\\
    \bottomrule
  \end{tabularx}
\end{table}

    \begin{figure}[!t]
             \centering
             \includegraphics[width= 3.5 in]{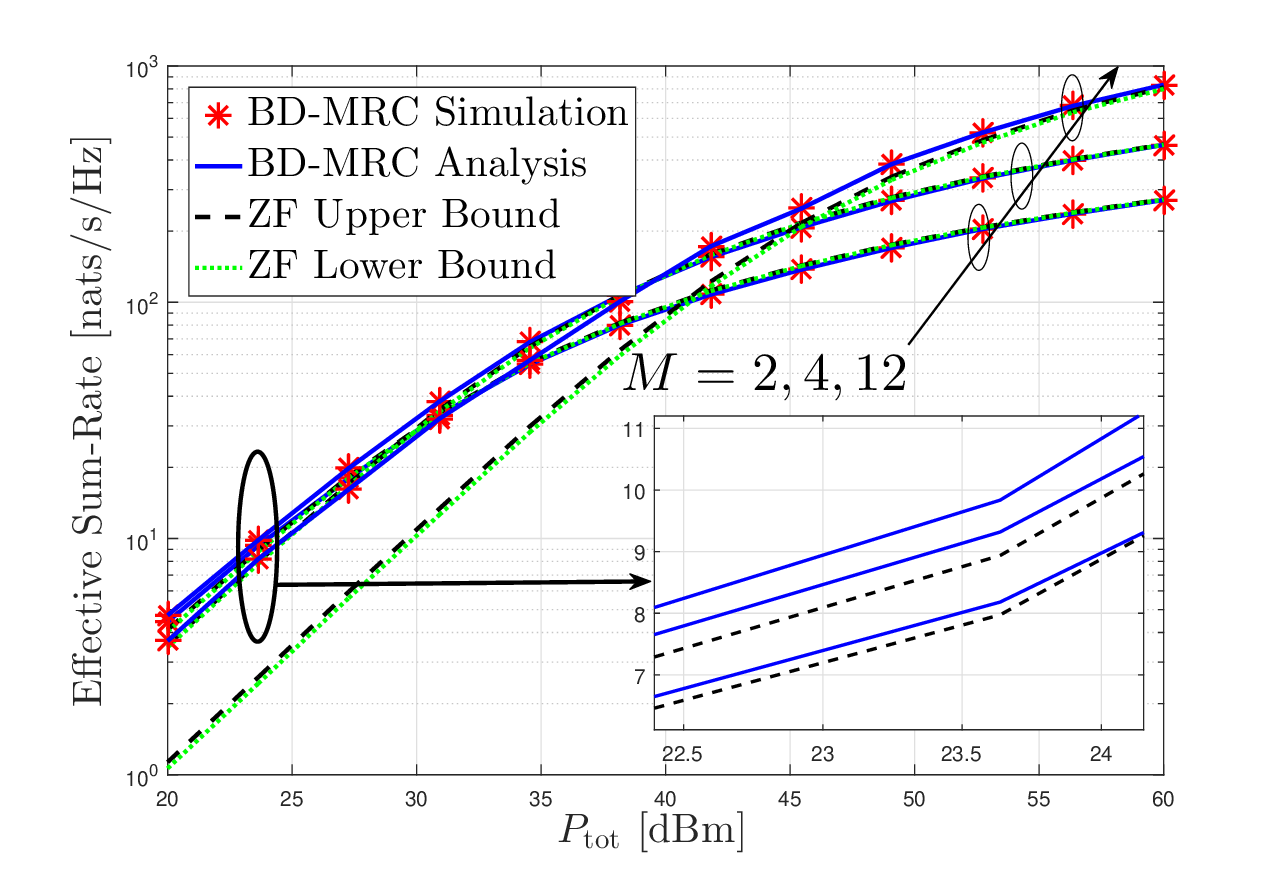}\vspace{-0.2cm}
             \captionsetup{font={footnotesize}}
		\caption{Effective sum-rate versus $P_{\rm tot}$ in  a Macro-cell under MMF with $L=64$, $G=5$, $Q=4$ and $J=M$.}\label{Sum_Rate_fig}
        \end{figure}

\subsection{Numerical Comparisons With Fixed $Q$ and $Q'$}

Fig. \ref{Sum_Rate_fig}  plots the effective sum-rate versus $P_{\rm tot}$ under the MMF power allocation. 
The analytical results represented by blue solid lines, derived from Theorem~\ref{R_BD_MRC_opt_Thm}, are obtained via a one-dimensional binary search for  $R^\star_{\text{BD-MRC}}$.
For the simulation results represented by red asterisk symbols in the same figure, we use the built-in function ``fminmax" in MATLAB to numerically solve the MMF optimization in \eqref{BD_MRC_Multiplex_Opt}. Interestingly, we note that the effective sum-rate increases with $M$ when SNR is higher, while it decreases with $M$ when SNR is relatively low, and where naturally spatial diversity is more impactful than multiplexing gain. For relatively small values of \( M \) (e.g., \( M \leq 4 \)), simple ZF precoding exhibits reasonably good performance, being only slightly inferior to BD-MRC across the entire considered SNR range. This observation aligns with the analysis in Remark~\ref{Small_Gap_ZFBD_remark}. However, for larger values of \( M \) (e.g., $M=12$), its performance in the low-SNR regime deteriorates significantly, leading to a substantial performance gap compared to BD-MRC in this region. 


     \begin{figure}[!t]
             \centering
             \includegraphics[width= 3.5 in]{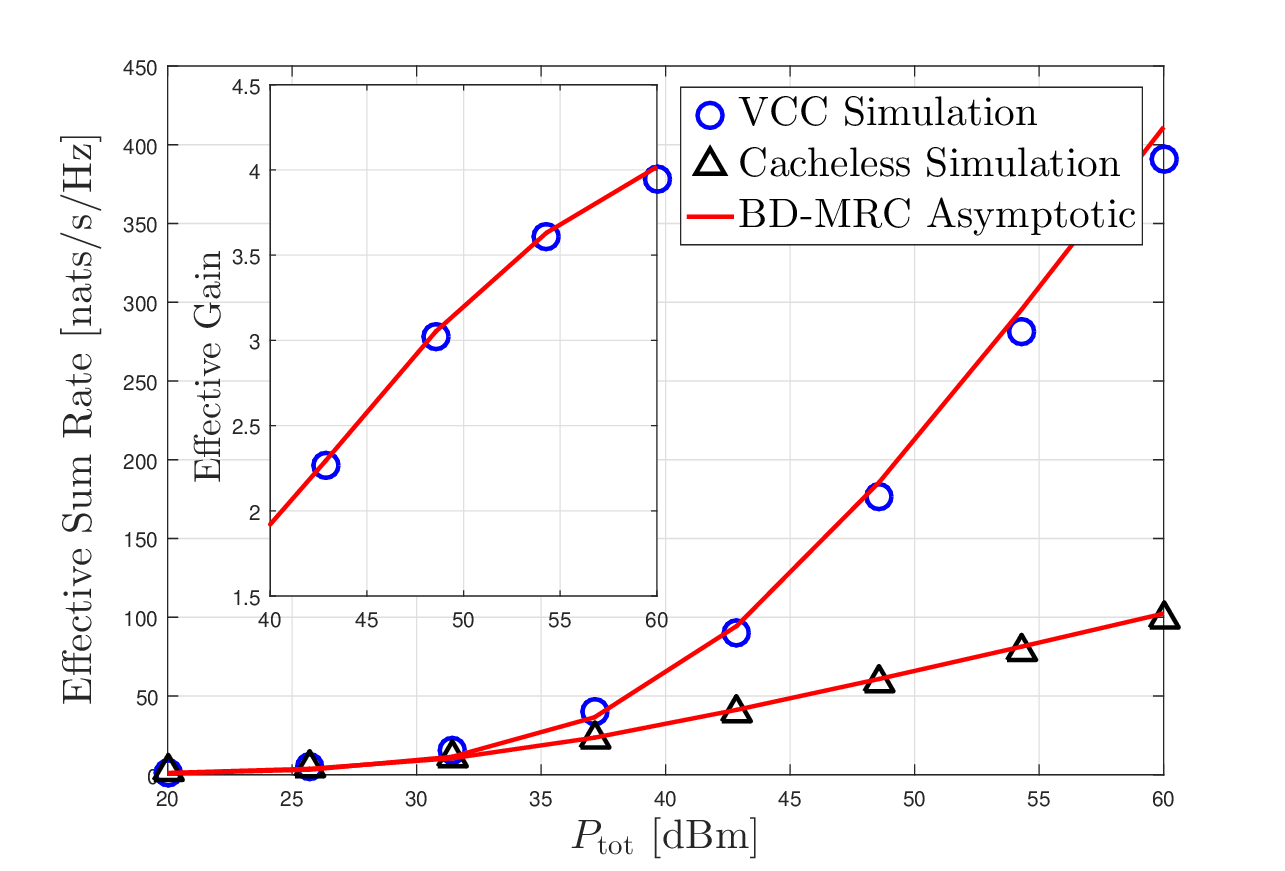}\vspace{-0.2cm}
             \captionsetup{font={footnotesize}}
		\caption{Effective sum-rate and effective gain  versus $P_{\rm tot}$ for $M=4$, $L=24$, $G=6$ and $Q = Q' =4$ in a Macro-cell under BD-MRC and MMF. Typical transmit power $P_{\rm tot}$ between 40--43 dBm.}\label{Gain_L128_1}
        \end{figure}	

     \begin{figure}[!t]
             \centering
             \includegraphics[width= 3.5 in]{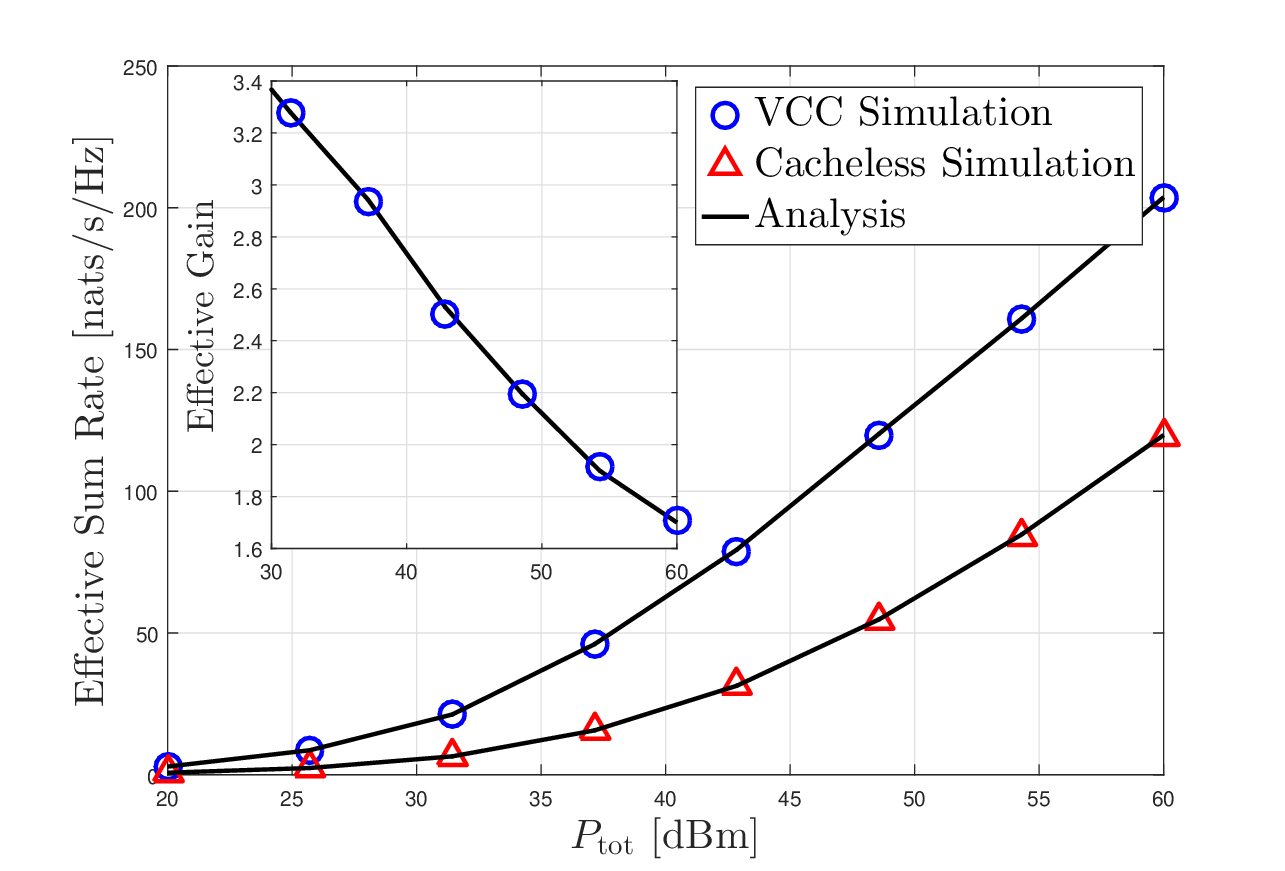}\vspace{-0.2cm}
             \captionsetup{font={footnotesize}}
		\caption{Effective sum-rate and effective gain versus $P_{\rm tot}$ for $L=32$, $M=4$, $Q=2$, $Q'=8$ and (a very conservative) $G=4$ in a Macro-cell under BD-MRC and MMF.}\label{Gain_L128_2}
        \end{figure}


Fig. \ref{Gain_L128_1} shows both the effective sum-rate and the corresponding gain versus $P_{\rm tot}$ for $G=6$, $L=24$, $M=4$ and $Q=Q'=4$ under BD-MRC and MMF in the Macro-cell. 
This setup is motivated by the assumption that the precoder size remains unchanged in both the cacheless and VCC cases. As seen from the VCC transmit signal in \eqref{signal_transmit_multi}, the system can sequentially precode \(G\) data vectors --- 
each of dimension \(QM \times 1\) --- and aggregate them for simultaneous transmission. In contrast, a conventional MU-MIMO system would transmit these vectors one after another. This structure allows VCC to support \(G\) times more users without requiring a larger precoder, making it compatible with existing transmitter architectures. In Fig. \ref{Gain_L128_1}, the simulated results are generated following the same procedure as in Fig.~\ref{Sum_Rate_fig}, and the red curve corresponds to the asymptotic expression derived in Lemma~\ref{BD_Lem1}. Even for moderate values of \(L\), the asymptotic result closely matches the simulation, demonstrating its high accuracy. In the effective gain plot, it is observed that for practical transmit power levels \(P_{\rm tot}\), the spectral efficiency is nearly \emph{doubled} compared to the conventional (cacheless) MU-MIMO system. Moreover, the effective gain continues to increase as \(P_{\rm tot}\) grows.


The delivery performance with the same DoF (i.e., $Q'=GQ$) in VCC and in the cacheless scenario under BD-MRC and MMF is plotted in Fig. \ref{Gain_L128_2}. 
In contrast to the increasing effective gain w.r.t. $P_{\rm tot}$ that we experience in Fig. \ref{Gain_L128_1}, the effective gain in Fig. \ref{Gain_L128_2} decreases with $P_{\rm tot}$, while still though allowing for good performance gains, especially in low SNR. This benefits from the fact that VCC introduces an extra multiplexing dimension for inter-user interference cancellation. This considerably alleviates the multiplexing load that should have been handled by multiple antennas, thus yielding greater beamforming gains for the users, which is particularly meaningful when requiring high beamforming gains in low SNR. 
Even at practical transmit power levels (\(P_{\rm tot} \in [40, 43]\) dBm), the effective gain achieved by VCC remains above 230\%, indicating a 1.3-fold improvement in spectral efficiency over the cacheless MU-MIMO system.

    \begin{figure}[!t]
             \centering
             \includegraphics[width= 3.5 in]{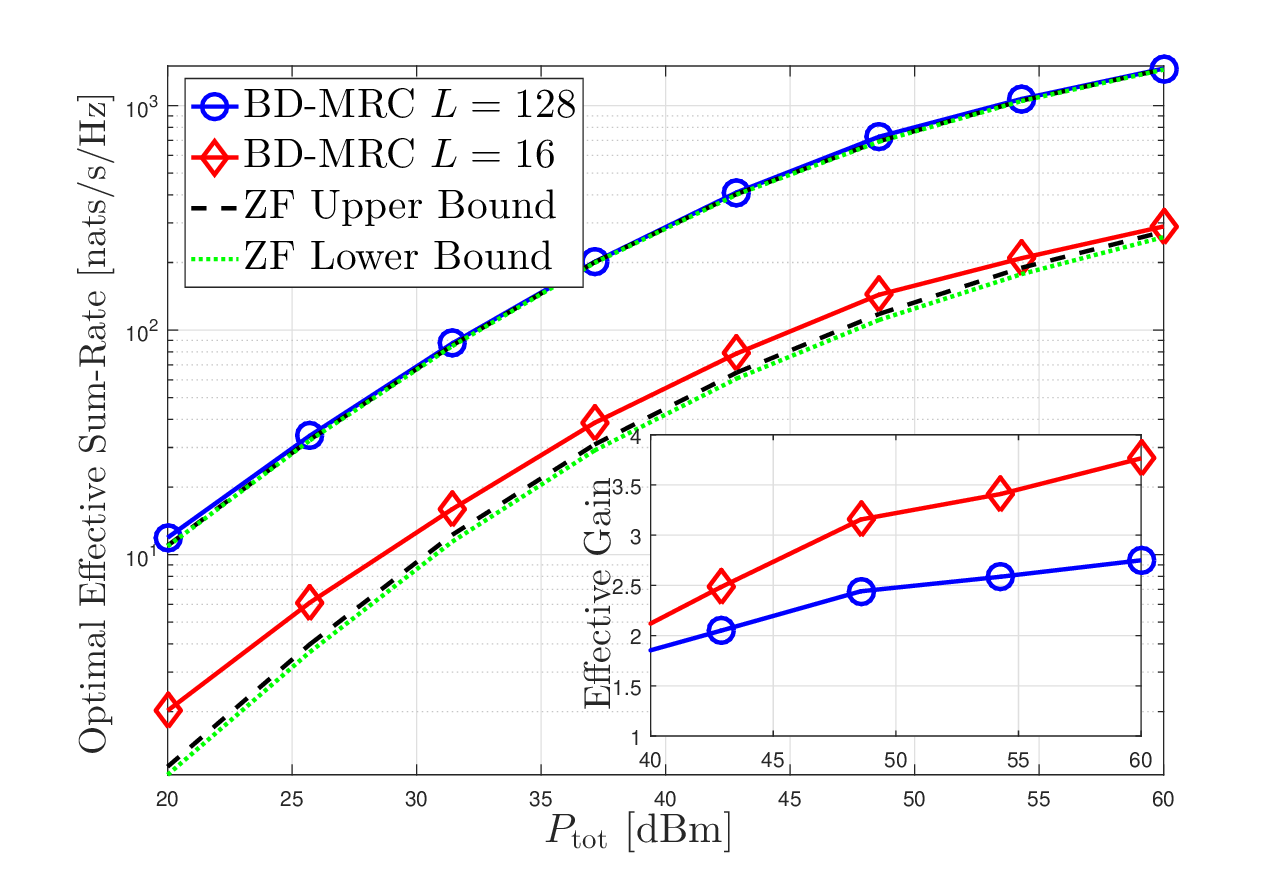}\vspace{-0.2cm}
             \captionsetup{font={footnotesize}}
		\caption{Optimal effective sum-rate and effective gain versus $P_{\rm tot}$ for $G=6$, $M=4$ in a Macro-cell under MMF, where $Q$ and $Q'$ are independently optimized.}\label{Opt_Gain_1_fig}
        \end{figure}

\subsection{Numerical Comparisons With Optimized $Q$ and $Q'$}			
Now we shift to the case of the optimal balance between multiplexing and beamforming gains, which will be shown to be particularly meaningful in the medium to high SNR regime.
Fig.~\ref{Opt_Gain_1_fig} plots the effective gain (corresponding to optimized rates --- where $Q$ and $Q'$ are independently optimized) versus $P_{\rm tot}$, for the Macro-cell case. 
While VCC provides a substantial gain over the corresponding traditional (cacheless) MU-MIMO scenario, this gain is lower than in the symmetric fading case (i.e., identical pathloss) studied in~\cite{Zhao_VectorCC}. This reduction is primarily due to the severe near-far effect, which significantly constrains the effective sum-rate --- particularly in large cells, even with optimal power allocation.\footnote{The situation is further exacerbated by the uniform distribution of users within the Macro-cell, with a majority positioned near the cell edge. For instance, $64.32\%$ of users are more than 300 meters from the BS, where, given a typical Macro-cell BS transmit power of $P_t=40$ dBm, their average SNR falls below 12.55 dB. Additional cases are detailed in Table 4.1 of \cite{EURECOM+7083}.} 
In contrast,
things are very different in the Micro-cell setting (see Fig.~\ref{VectorMulti_Micro_fig}) where the effective gain of optimized rates is again very notable. For example, in a Micro-cell setting, the recorded gain is $410\%$ for the reasonable BS transmit power of $P_{\rm tot}=33$ dBm. It is also worth noting that simple ZF precoding --- with power allocation based on pathloss among the served users --- yields a tight lower bound for the BD-MRC scheme when \(L\) is large and \(M\) is relatively small, as observed in Figs.~\ref{Opt_Gain_1_fig}--\ref{VectorMulti_Micro_fig}. This observation is consistent with the analytical result in Remark~\ref{Small_Gap_ZFBD_remark}.

     \begin{figure}[!t]
             \centering
             \includegraphics[width= 3.5 in]{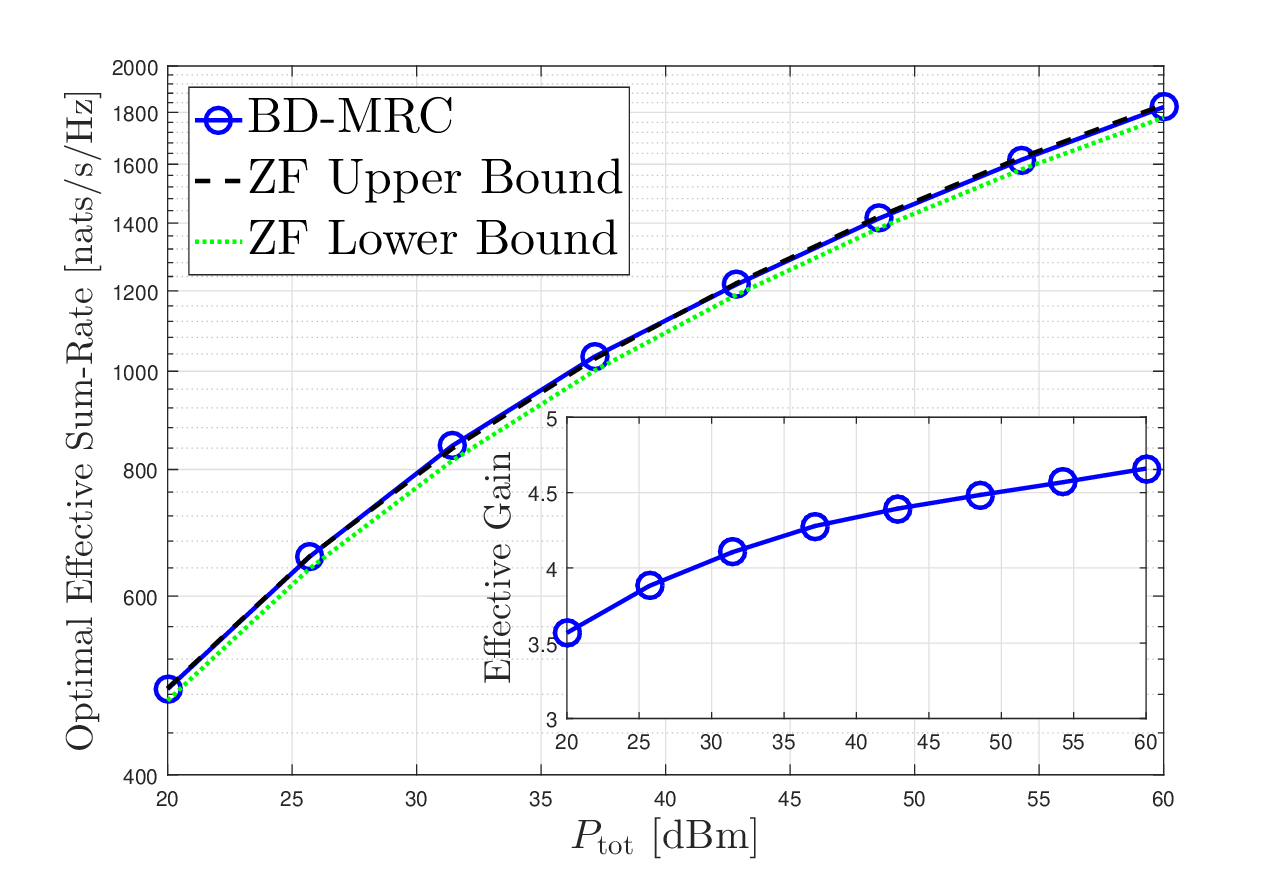}\vspace{-0.2cm}
             \captionsetup{font={footnotesize}}
		\caption{Optimal effective sum-rate and effective gain versus $P_{\rm tot}$ for $L=32$ $M=2$, and $G=6$ in a Micro-cell under MMF, where $Q$ and $Q'$ are independently optimized.}\label{VectorMulti_Micro_fig}
        \end{figure}

\subsection{Comparisons to Bit-Level Multi-Antenna Coded Caching}\label{Com_BCC_Subsec}
Since coded caching is essentially a form of network coding, most of VCC schemes (e.g., \cite{Lampiris_JSAC,Bakhshzad,Emanuele,Elizabath,Ting,Zhao_VectorCC}) admit a PHY representation consistent with~\eqref{signal_transmit_multi}, characterized by a superposition of precoded $L \times 1$ signal streams. In contrast, the conventional multi-antenna framework corresponds to \emph{bit-level coded caching} (BCC) (e.g.,~\cite{Ngo,Shariatpanahi,Tolli_TWC}), where subfiles intended for a given user set are combined through bitwise XOR operations. This motivates a comparison between the PHY delivery performance of VCC and BCC. However, as noted in the Introduction, most existing multi-antenna coded caching schemes---typically bit-level---ignore the practical subpacketization constraint, which limits their achievable DoF to approximately $\Lambda\gamma + Q$. Under finite file sizes,~\cite{9163148} revisits the canonical multi-server coded caching framework~\cite{Shariatpanahi,Tolli_TWC} and proposes a reduced-complexity variant that preserves its main properties. We therefore adopt this \emph{bit-level multi-server variant (MSV)} from~\cite{9163148} as the comparison baseline for VCC, and for simplicity, single-antenna users are considered, as also assumed in~\cite{9163148}. Furthermore, symmetric Rayleigh fading is assumed for all served users.

In the MSV scheme~\cite{9163148}, the transmitter generates \(L\) signal streams, among which \(L-1\) are dedicated to \(L-1\) unicast users and one stream carries an XOR-coded (common) message intended for \emph{another} \(G = \Lambda\gamma + 1\) users. Once the XOR message is decoded, each of these \(G\) users can reconstruct its desired subfile using the cached content. Under the equal power allocation policy, the transmit signal is given by
\begin{align}\label{transmit_MSV}
    {\bf x}_{\rm MSV} = \sqrt{\frac{P_{\rm tot}}{L}} {\bf f}_0 c_0 + \sum_{k=1}^{L-1} \sqrt{\frac{P_{\rm tot}}{L}} {\bf f}_k c_k,
\end{align}
where \(c_0\) represents the common XOR-coded message, \(c_k\) denotes the symbol for the \(k\)-th unicast user, and \({\bf f}_0, \ldots, {\bf f}_{L-1}\) are unit-norm precoding vectors of size \(L\times1\). The vector \({\bf f}_0\) is designed to cancel the interference from the unicast signals at a \emph{specific} multicast user (e.g., the first user), enabling that user to decode the XOR message. The remaining \(G-1\) multicast users exploit their cached content to remove unicast interference before decoding the same XOR message. Consequently, the scheme simultaneously serves \(L-1 + G = L + \Lambda\gamma\) users, which corresponds to its achievable DoF. Additional design details of the MSV framework can be found in~\cite{9163148}.  

Let \({\bf h}_{{\rm uc},k} \in \mathbb{C}^{L\times1}\) (\(k \in [L-1]\)) and \({\bf h}_{{\rm mc},k'} \in \mathbb{C}^{L\times1}\) (\(k' \in [G]\)) denote the channels of unicast and multicast users, respectively. The precoding vectors in~\eqref{transmit_MSV} are designed as  
\begin{align}
   & {\bf f}_0 \in \mathbb{O}\{ {\bf h}_{{\rm uc},1}, {\bf h}_{{\rm uc},2}, \ldots, {\bf h}_{{\rm uc},L-1} \}, \notag\\
   & {\bf f}_k \in \mathbb{O}\{ {\bf h}_{{\rm mc},1}, {\bf h}_{{\rm uc},1}, \ldots, {\bf h}_{{\rm uc},k-1}, {\bf h}_{{\rm uc},k+1}, \ldots, {\bf h}_{{\rm uc},L-1} \}, \notag
\end{align}
which can be implemented using BD precoding.

The transmission rate of the XOR-coded message is  
\begin{align}\label{multicast_rate}
    R_{\rm mc} = G \min_{k' \in [G]} \ln \!\left( 1 + \frac{P_{\rm tot}}{N_0 L} \big| {\bf h}_{{\rm mc},k'} {\bf f}_0 \big|^2 \right),
\end{align}
where the minimum operator ensures successful decoding of the multicast message among all \(G\) target users.  
The transmission rate of the \(k\)-th unicast stream is  
\begin{align}
    R_{{\rm uc},k} = \ln \!\left( 1 + \frac{P_{\rm tot}}{N_0 L} \big| {\bf h}_{{\rm uc},k} {\bf f}_k \big|^2 \right).
\end{align}

Since the spatial multiplexing resources are fully utilized, no additional dimension is available for beamforming toward individual users, leading to  
\begin{align}
    & {\bf h}_{{\rm mc},k'} {\bf f}_0 \sim \mathcal{CN}(0,1),~\forall k' \in [G], \\
    & {\bf h}_{{\rm uc},k} {\bf f}_k \sim \mathcal{CN}(0,1),~\forall k \in [L-1].
\end{align}
The effective total rate of the MSV scheme is of the form 
\begin{align}
    R_{\rm MSV}  = \xi_{L+\Lambda\gamma} \!\left( R_{\rm mc} + \sum_{k=1}^{L-1} R_{{\rm uc},k} \right),
\end{align}
where \(\xi_{L+\Lambda\gamma} \triangleq 1 - \Theta (L+\Lambda\gamma)/T\) represents the CSI overhead factor.
Considering BD precoding in the cacheless baseline, the effective gain of MSV is defined as  
\begin{align}
    \mathcal{G}_{\rm MSV} = 
    \frac{\mathbb{E}\{ R_{\rm MSV} \}}
    {\max_{Q' \in [L]} \mathbb{E}\{ R_{\rm BD}(1,Q') \}},
\end{align}
where the expectation is taken over channel realizations, and \(R_{\rm BD}(1,Q')\) denotes the effective rate of BD precoding analyzed in Section~\ref{multi_ana_sec} for single-antenna receivers under symmetric channel conditions (\(\beta_k = 1,~\forall k \in [K]\)).  
We note that the original MSV scheme in \cite{9163148} does not adjust the number of unicast streams in~\eqref{transmit_MSV} to optimize the multiplexing–beamforming tradeoff, as in VCC. To fully exploit the transmission capability of the MSV scheme under finite SNR, we further adjust the number of unicast streams \(Q_{\rm uc} \le L-1\) in~\eqref{transmit_MSV} to achieve an optimized multiplexing--beamforming tradeoff, referred to as the \emph{modified MSV scheme}, where the optimal BD design follows the same procedures as we considered in VCC (cf. Section~\ref{BD_MRC_sub}). The corresponding effective gain is defined as  
\begin{align}
\mathcal{G}_{\rm MSV}' \triangleq  
\frac{ \max_{Q_{\rm uc} \in [L-1]} \mathbb{E}\{ R_{\rm MSV} (Q_{\rm uc}) \} }
     { \max_{Q' \in [L]} \mathbb{E}\{ R_{\rm BD}(1,Q') \} }.
\end{align}

As \(P_{\rm tot} \to \infty\), the optimal \(Q'\) in the cacheless BD baseline equals \(L\), leading to the high-SNR limit for MSV as
\begin{align}
    \mathcal{G}_{\rm MSV} \longrightarrow \frac{L + \Lambda\gamma}{L},
\end{align}
which approaches unity as \(L \to \infty\). Even for moderate values, such as \(L=32\) and \(G=6\), the high-SNR limit gain remains marginal, which equals \(\frac{32+5}{32}=1.1563\) as shown in Fig.~\ref{MSV_VCC_fig}.

Fig.~\ref{MSV_VCC_fig} compares the effective gains of MSV and BD-based VCC against the cacheless BD baseline. Both MSV and VCC exhibit increasing gains with SNR, yet VCC shows a notably sharper rise, implying a growing advantage over the cacheless baseline. In contrast, the effective gain of original MSV in \cite{9163148} remains below 1 at finite SNRs, indicating that the cacheless baseline consistently outperforms original MSV. The main reasons are threefold:  
$i)$ the DoF limitation of MSV, equal to \(L+\Lambda\gamma\), providing only marginal gain for large \(L\) at finite SNRs;  
$ii)$ the inherent worst-user bottleneck in multicasting, where the XOR rate is constrained by the weakest user (cf. \eqref{multicast_rate}); and  
$iii)$ the additional CSI acquisition required for approximately \(\Lambda\gamma\) more users compared with the cacheless case, which further reduces the effective throughput. As illustrated in Fig.~\ref{MSV_VCC_fig}, while the modified MSV scheme achieves the optimal multiplexing–beamforming tradeoff at finite SNRs, its effective gain, though clearly improved, remains close to one, implying limited advantage over the cacheless baseline. This result corroborates the claim made in the Introduction that the caching gains achieved by bit-level multi-antenna coded caching are overshadowed by the conventional multiplexing gains.

     \begin{figure}[!t]
             \vspace{-0.2cm}
             \centering
             \includegraphics[width= 3.5 in]{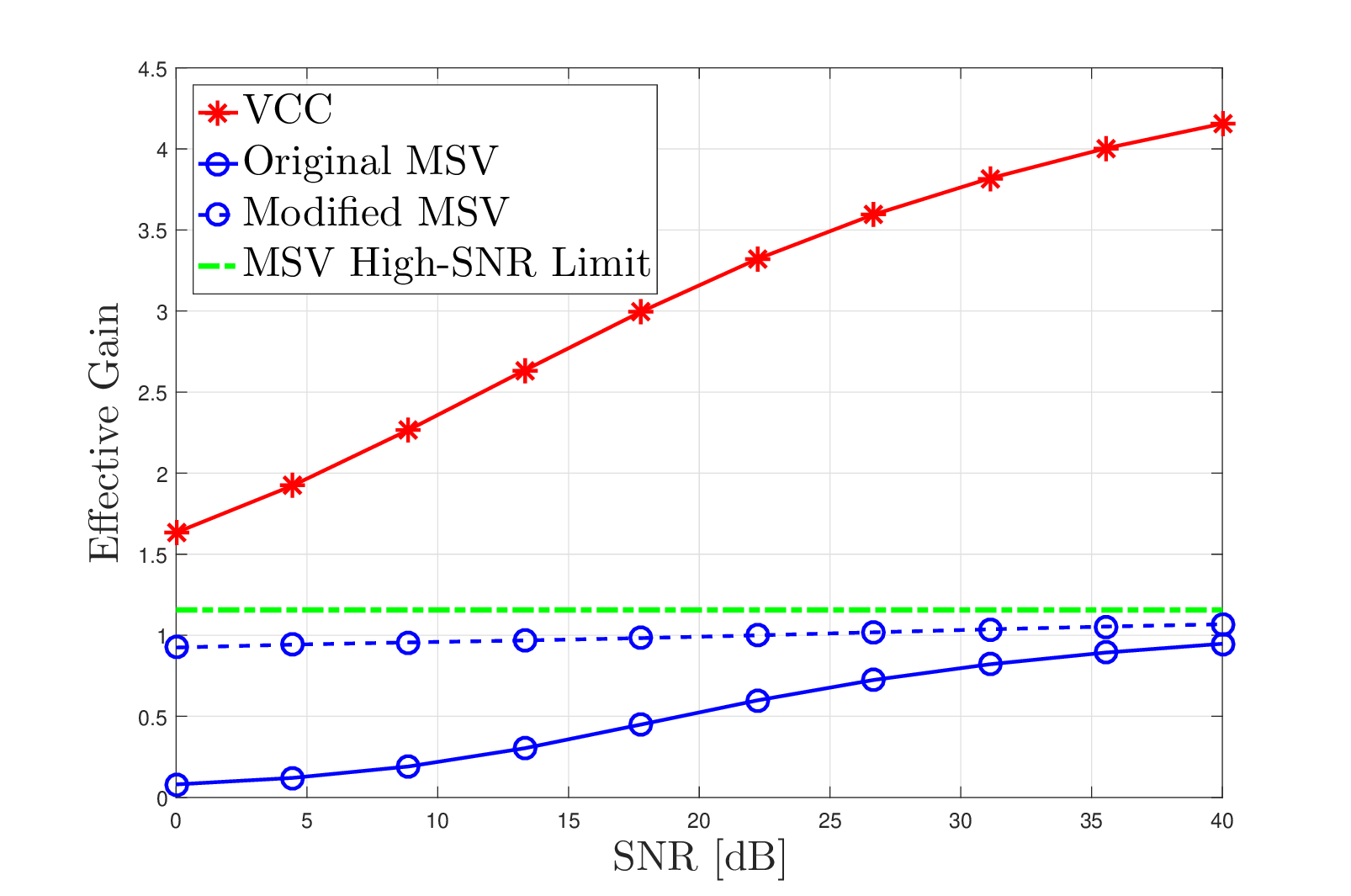}\vspace{-0.2cm}
             \captionsetup{font={footnotesize}}
		\caption{Effective gain versus ${\rm SNR} \triangleq P_{\rm tot}/N_0$ for $L=32$ and $G=6$ over symmetric Rayleigh fading channels}\label{MSV_VCC_fig}
        \end{figure}

\subsection{VCC With Imperfect CSIT and CSIR}
In this subsection, we present two parts. The first part considers the case with imperfect CSIT while assuming perfect CSI at the receiver (CSIR). The second part extends the study to include both imperfect CSIT and CSIR.

\subsubsection{Imperfect CSIT}
We now turn to the case of imperfect CSIT --- a common scenario in practice due to pilot contamination, estimation errors, and feedback delays. While accurate CSI typically enhances spectral efficiency, Fig.~\ref{CIST_fig1} illustrates that VCC can, in some cases, achieve even higher gains under imperfect CSIT than in the perfect CSIT setting. This is because inter-group interference can be mitigated using cached content, independent of CSIT quality --- a sharp contrast to conventional cacheless MU-MIMO, which relies heavily on accurate CSIT for interference management. Throughout, we assume perfect CSIR,\footnote{CSIR is typically more accurate than CSIT, as it is estimated directly from pilots at each receiver. In contrast, CSIT acquisition --- via feedback in FDD or uplink pilots in TDD --- is more error-prone due to delay, quantization, interference, and user coordination. This disparity further amplifies the relative gains of VCC.} reflecting its more robust acquisition in practical systems. 
To simplify the analysis, we focus on single-antenna receivers under ZF precoding, aiming to assess whether the reported gains persist despite CSIT imperfections.

With single-antenna receivers, the channel between the BS and user \( {\rm U}_{\psi,k} \) (for any \( \psi \in \Psi \) and \( k \in [Q] \)) is now represented as a vector, denoted by \( {\bf h}_{\psi,k} \in \mathbb{C}^{L} \).
Let $\hat {\bf h}_{\psi,k}$ denote the CSI estimate at the BS. Considering the linear MMSE estimation \cite{Hien_TCOM}, we can express $ {\bf h}_{\psi,k}$ as
$
 {\bf h}_{\psi,k} = \hat {\bf h}_{\psi,k} + \tilde {\bf h}_{\psi,k},
$
where $\tilde {\bf h}_{\psi,k}$ represents the estimation errors, whose elements follow the i.i.d. Gaussian distribution with zero-mean and variance $\tilde \beta_{\psi,k}$. Naturally, $\tilde {\bf h}_{\psi,k}$ is independent of the MMSE estimate $\hat {\bf h}_{\psi,k}$.

The BS treats $\hat {\bf h}_{\psi,k}$ as the real channel vector and then follows the same paradigm as in \eqref{ZF_Design_Multi} to perform ZF precoding. Let $\hat {\bf V}_\psi$ denote the ZF matrix based on the MMSE estimates $\{  \hat {\bf h}_{\psi,k}: {k \in [Q]} \}$. Following the channel training method in \cite{caire2010multiuser}, after the BS finishes the CSI estimation, the BS issues an additional downlink training phase, where it sends several orthogonal training symbols along each of the precoding vectors to user ${\rm U}_{\psi,k}$ for any $\psi \in \Psi$ and $k \in [Q]$ for coherent detection. By doing so, ${\rm U}_{\psi,k}$ can obtain the coupling (composite) coefficients $\{ \sqrt{P_{\phi,k'}} {\bf h}_{\psi,k}^T \hat {\bf v}_{\phi,k'}: \phi \in \Psi, k' \in [Q] \}$, where ${\bf v}_{\phi,k'}$ denotes the $k'$-th column of $\hat {\bf V}_\phi$. As the downlink training is assumed to be perfect, ${\rm U}_{\psi,k}$ detects these coupling coefficients perfectly (perfect composite CSIR), which enables it to completely cancel the inter-group interference by applying its cached content, as described in Section~\ref{sys_sec}. After removing the inter-group interference, the remaining signal at ${\rm U}_{\psi,k}$ is
\begin{align}\label{CSIT_signal}
    y_{\psi,k}'
     =&\sqrt{P_{\psi,k}} {\bf h}_{\psi,k}^T \hat {\bf v}_{\psi,k}  s_{\psi,k} + z_{\psi,k} \notag\\
     &+ \sum\nolimits_{k' \in [Q] \setminus k  } \sqrt{P_{\psi,k'}} {\bf h}_{\psi,k}^T \hat {\bf v}_{\psi,k'} s_{\psi,k'} ,
\end{align}
yielding a SINR when decoding $s_{\psi,k}$ of the form
\begin{align}\label{SINR_ZF_iCSIT}
    \text{SINR}_{\psi,k}^{\rm ZF} = \frac{ P_{\psi,k} | {\bf h}_{\psi,k}^T \hat {\bf v}_{\psi,k} |^2 }{N_0 + \sum\nolimits_{k' \in [Q] \setminus k  } P_{\psi,k'} | {\bf h}_{\psi,k}^T \hat {\bf v}_{\psi,k'} |^2},
\end{align}
and an effective sum rate, under imperfect CSIT, of the form
\begin{align}\label{rate_ZF_csit}
    R_{\rm sum}^{\rm ZF} (G,Q) \!=\! \left(\! 1 \!-\! \frac{G Q\Theta}{T} \!\right) \!\sum_{\psi \in \Psi} \sum_{k\in [Q]} \! \ln\! \left(1 \!+\! \text{SINR}_{\psi,k}^{\rm ZF} \right),
\end{align}
where $\text{SINR}_{\psi,k}^{\rm ZF}$ is given by \eqref{SINR_ZF_iCSIT}. The effective sum rate of the corresponding cacheless case is $R_{\rm sum}^{\rm ZF} (1,Q') $, where $Q'$ is the operational multiplexing gain under conventional ZF precoding. As shown in \cite{EURECOM+7083}, power allocation for MMF optimization marginally improves the effective gain, despite the notable increase in the effective sum-rate. For simplicity, we consider equal power allocation among the users and then consider the effective gain as defined below
\begin{align}\label{ZF_gain_imperfect}
    \mathcal{G}_{\rm ZF}' \triangleq \frac{\max_{Q \in [L]} \mathbb{E}\{R_{\rm sum}^{\rm ZF} (G,Q)\} }{\max_{Q' \in [L]} \mathbb{E}\{R_{\rm sum}^{\rm ZF} (1,Q')\}},
\end{align}
where the expectation is averaged over channel states.

  \begin{figure}[!t]
             \centering
              \vspace{-0.3cm}
             \includegraphics[width= 3.7 in]{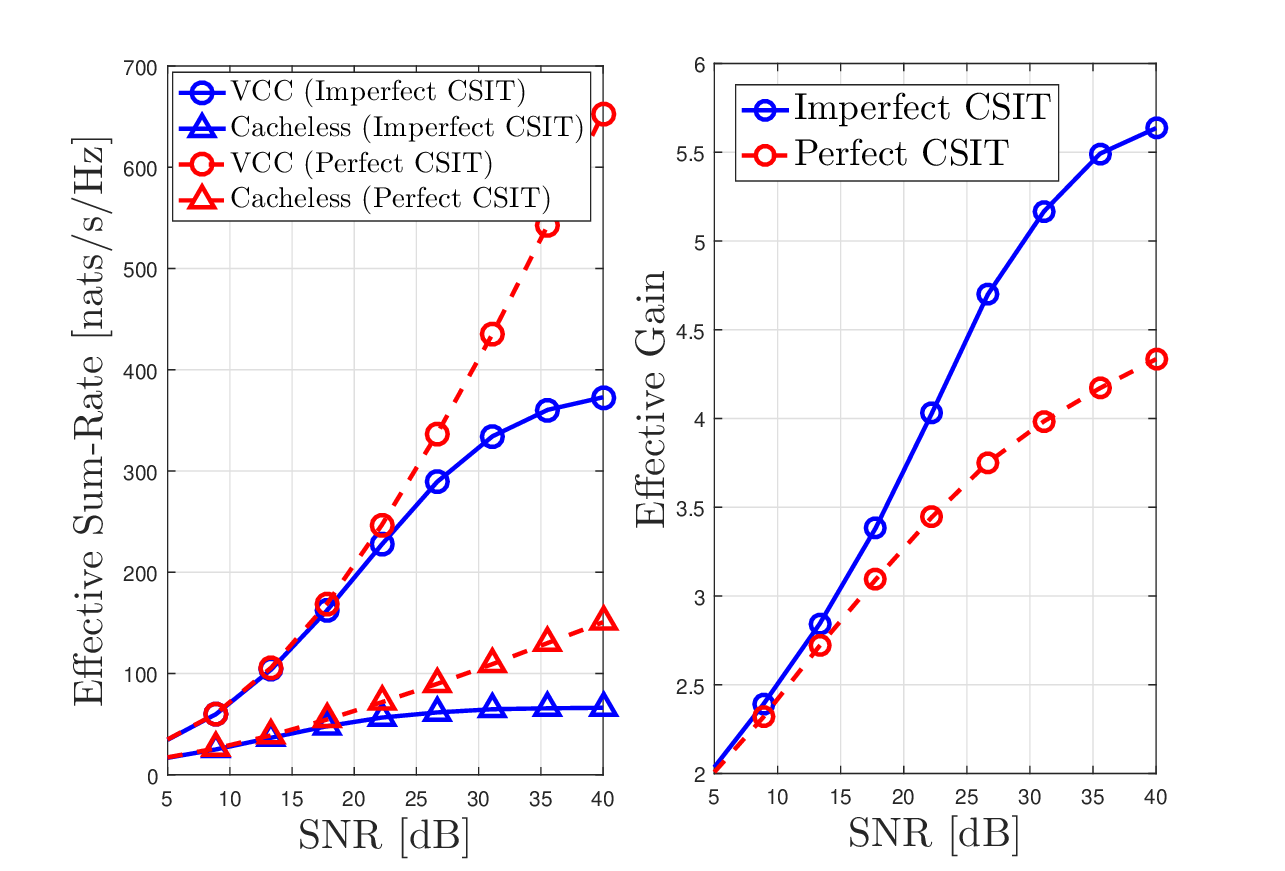}\vspace{-0.2cm}
             \captionsetup{font={footnotesize}} 
		\caption{Effective sum-rate and effective gain versus SNR for $L=16$, $M=1$ and $G=6$ under ZF precoding, where the users are statistically symmetric. Specifically, $\beta_{\psi,k}=1$, and $\tilde \beta_{\psi,k}= \tilde \beta =0.01$ for any served user in the case of imperfect CSIT.  $\text{SNR}$ denotes the average transmit SNR $\triangleq P_{\rm tot}/N_0$.}\label{CIST_fig1}
        \end{figure}



Fig. \ref{CIST_fig1} is obtained by selecting the optimal multiplexing gain for each SNR value, while the effective gain is computed according to \eqref{ZF_gain_imperfect}. As expected, under imperfect CSIT, VCC exhibits lower effective sum-rate than the perfect CSIT case in the medium-to-high SNR regime, since the system is interference-limited due to CSIT estimation errors. Surprisingly, however, the effective gain achieved by VCC under imperfect CSIT surpasses that of the perfect CSIT case in the medium-to-high SNR region. This phenomenon is primarily attributed to the fact that, in VCC, inter-group interference is completely canceled without relying on CSIT, significantly mitigating the impact of CSIT errors on system performance. In contrast, for the cacheless MU-MIMO system, the elimination of inter-user interference heavily depends on the quality of the CSIT, making its performance more sensitive to CSIT imperfections.



\subsubsection{Imperfect CSIT and CSIR}
In the following, we examine the impact of imperfect CSIT and CSIR on VCC delivery. Let \(A_{\psi,k}^{\phi,\vartheta} \triangleq {\bf h}_{\psi,k}^T \hat{\bf v}_{\phi,\vartheta}\) denote the coupling coefficient for all \(\psi,\phi \in \Psi\) and \(k,\vartheta \in [Q]\). Under MMSE estimation at the receiver,
\(A_{\psi,k}^{\phi,\vartheta} = \hat A_{\psi,k}^{\phi,\vartheta} + \tilde A_{\psi,k}^{\phi,\vartheta}\),
where \(\hat A_{\psi,k}^{\phi,\vartheta}\) is the estimate and \(\tilde A_{\psi,k}^{\phi,\vartheta} \sim \mathcal{CN}(0,\tilde\beta')\) is the error. For simplicity, we consider a symmetric case with identical error variance \(\tilde\beta'\) for all coefficients.

With imperfect CSIR, cache-aided interference cancellation is \emph{not perfect} and leaves \emph{residual interference}. The signal after imperfect cache-aided interference cancellation is
\begin{align}
    y_{\psi,k}'' \;=\; y_{\psi,k}' \;+\; \sum_{\phi \in \Psi \setminus \psi} \sum_{\vartheta \in [Q]}
    \sqrt{P_{\phi,\vartheta}}\, \tilde A_{\psi,k}^{\phi,\vartheta}\, s_{\phi,\vartheta},
\end{align}
where \(y_{\psi,k}'\) is the signal for the imperfect-CSIT-only case (see \eqref{CSIT_signal}). Assuming equal power allocation and ZF precoding, the resulting SINR is
\begin{align}\label{SINR_csir}
    \mathrm{SINR}_{\psi,k}^{\rm ZF}
    \;=\;
    \frac{ \frac{P_{\rm tot}}{GQ}\, \big| \hat A_{\psi,k}^{\phi,\vartheta} \big|^2 + \frac{P_{\rm tot}}{GQ}\, \tilde\beta'}
         { N_0 + I },
\end{align}
where the interference term \(I\) accounts for residuals due to \emph{both} imperfect CSIT and CSIR:
\begin{align} 
I &\triangleq \frac{P_{\rm tot}}{GQ} \sum_{k' \in [Q] \setminus k } \! \! \mathbb{E}\big\{ | \tilde A_{\psi,k}^{\psi,k'} |^2 \big\} + \frac{P_{\rm tot}}{GQ} \sum_{\phi \in \Psi \setminus \psi} \sum_{\vartheta \in [Q]} \mathbb{E}\big\{ | \tilde A_{\psi,k}^{\phi,\vartheta} |^2 \big\} \notag\\ & = \frac{P_{\rm tot}}{GQ} \tilde \beta' (GQ-1). 
\end{align}
Substituting \eqref{SINR_csir} into \eqref{rate_ZF_csit} and \eqref{ZF_gain_imperfect} yields the effective gain under simultaneous CSIT and CSIR imperfections.

Fig.~\ref{CISR_fig} illustrates the effective gain of VCC under both imperfect CSIT and CSIR conditions with varying CSIR estimation accuracy. 
At low SNR, the effective gains under different CSIR qualities are close to that of the perfect-CSI case, consistent with the observation in Fig.~\ref{CIST_fig1}. 
As the SNR increases, however, residual inter-user interference caused by imperfect CSIT and CSIR becomes more significant, leading to a gradual reduction in the effective gain. 
In practical systems, the CSIR estimation accuracy is usually much higher than that of CSIT, i.e., $\tilde{\beta}' \ll \tilde{\beta}$, where $\tilde \beta$ denotes the variance of CSIT errors.
Therefore, the case of $\tilde{\beta}' = \tilde{\beta}$ (green curve) represents a worst-case scenario. 
Even in this worst case, VCC still achieves \emph{more than triple} the spectral efficiency of the cacheless baseline when SNR exceeds 20~dB. 
The numerical comparisons in Fig.~\ref{CISR_fig} confirm that VCC maintains a substantial multiplicative spectral-efficiency gain over cacheless MU-MIMO systems, even in the presence of simultaneous CSIT and CSIR imperfections.

  \begin{figure}[!t]
             \centering
              \vspace{-0.3cm}
             \includegraphics[width= 3.5 in]{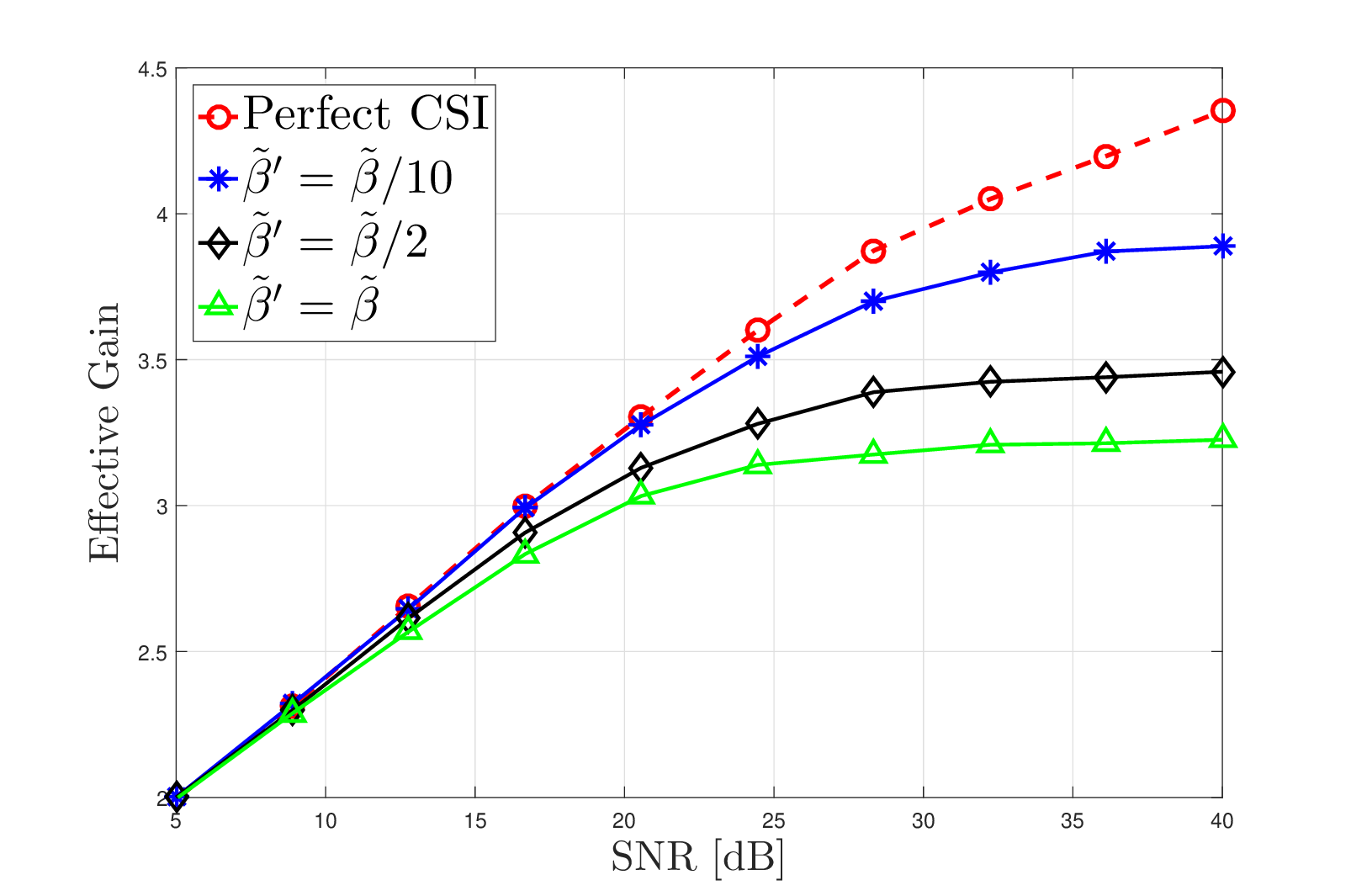}\vspace{-0.2cm}
             \captionsetup{font={footnotesize}} 
		\caption{Effective gain versus SNR under the same parameters as Fig.~\ref{CIST_fig1}, with additional imperfect CSIR consideration. }\label{CISR_fig}
        \end{figure}

\section{Conclusions and Future Work}\label{conclude_sec}
We investigated VCC under realistic conditions, including variable pathloss, multi-antenna receivers, MMF, and practical precoders and combiners. We derived analytical expressions for the effective sum-rate and gain under BD-MRC and ZF schemes, along with several closed-form results under simplifying assumptions. Numerical evaluations confirmed the tightness of these expressions and demonstrated the significant gains of VCC in realistic scenarios. These gains are particularly pronounced in Micro-cell environments, where we observed spectral efficiency improvements exceeding 300\% over conventional (cacheless) MU-MIMO systems. Moreover, Fig.~\ref{MSV_VCC_fig} compared the effective gains of VCC with the multi-server scheme and its variants, demonstrating its performance advantage over conventional multi-antenna coded caching approaches. Importantly, all reported gains in Figs.~\ref{Opt_Gain_1_fig}--\ref{CISR_fig} represent improvements over optimized traditional MU-MIMO baselines, including any additional CSI acquisition costs.

In addition to the obvious multiplexing dimension for interference cancellation enabled by additional memory, VCC enhances beamforming gains which are especially valuable in low-SNR regimes. The technique also remains effective even under constraints such as a limited number of active signal streams (e.g., $GQM \le 32$ in Fig~\ref{Gain_L128_2}). Beyond its theoretical advantages, VCC offers strong practical value. With VoD representing over 70\% of network traffic and cacheable content comprising nearly 90\% of data consumption, such caching-based strategies present a scalable solution to enhance spectral efficiency and reduce congestion in high-demand, content-heavy networks.

It is worth noting that the simulation results presented are somewhat conservative. In particular, the performance gains of the proposed cache-aided MU-MIMO scheme become significantly more pronounced when imperfect CSIT is taken into account. This is a key strength of our approach: its inherent robustness and improved performance under CSIT imperfections, by virtue of the fact that VCC shifts much of the interference management to the receivers, inherently reducing per-user CSIT needs. A preliminary indication of these enhanced gains was illustrated in Fig.~\ref{CIST_fig1}. Future work may focus on extending the current framework to explicitly characterize the impact of imperfect CSIT and CSIR. In particular, it would be valuable to analyze the performance under different CSI quality levels and investigate the trade-offs among CSI accuracy, feedback overhead, and system throughput. When VCC is applied in large-scale systems, several scalability challenges arise as the numbers of users and antennas increase. These include higher CSI overhead, power-efficiency concerns, and increased computational complexity; see \cite{EURECOM+7083} for further discussion. In addition, user asynchronization---where users reveal their requests at different time instants---can create asymmetric signal structures during the delivery phase, leading to different dimensions of ${\bf s}_\psi$ for various $\psi \in \Psi$ in \eqref{signal_transmit_multi}. A more comprehensive study of these issues, which forms a substantial research topic on its own, is deferred to future work.


\appendices
\renewcommand{\thesectiondis}[2]{\Roman{section}:}

\section{Proof of Lemma \ref{BD_MRC_design_lem}}\label{Proof_BD_MRC_design_lem}    
We will find ${\bf m}_{\psi,k,q}$ by decomposing a much smaller sized matrix $ {\bf H}_{\psi,k}^T {\bf T}_{\psi,-k}^2  {\bf H}_{\psi,k}^*$ $= {\bf H}_{\psi,k}^T {\bf T}_{\psi,-k}  {\bf H}_{\psi,k}^*$. 
With $\lambda_{\psi,k,q}$ and ${\bf t}_{\psi,k,q}$ in place, we have
\begin{align}\label{HTH_lambdat_eq}
    {\bf H}_{\psi,k}^T {\bf T}_{\psi,-k}  {\bf H}_{\psi,k}^* {\bf t}_{\psi,k,q} = \lambda_{\psi,k,q} {\bf t}_{\psi,k,q}.  
\end{align}
   Multiplying ${\bf T}_{\psi,-k} {\bf H}_{\psi,k}^*$ on both sides of \eqref{HTH_lambdat_eq} and using that ${\bf T}_{\psi,-k} = {\bf T}_{\psi,-k}^2$ yields 
    \begin{align}
    &{\bf T}_{\psi,-k}  {\bf H}_{\psi,k}^* {\bf H}_{\psi,k}^T {\bf T}_{\psi,-k} \big( {\bf T}_{\psi,-k}  {\bf H}_{\psi,k}^* {\bf t}_{\psi,k,q} \big)  \notag\\
    &\hspace{3cm}= \lambda_{\psi,k,q} \big( {\bf T}_{\psi,-k}  {\bf H}_{\psi,k}^* {\bf t}_{\psi,k,q}  \big), \label{eigen_vec_Omega} 
    \end{align}
which shows that both ${\bf T}_{\psi,-k}  {\bf H}_{\psi,k}^* {\bf t}_{\psi,k,q} $ and $ {\bf m}_{\psi,k,q}$ are the eigenvectors associated to the $q$-th largest eigenvalue $\lambda_{\psi,k,q}$ of $ {\bf T}_{\psi,-k} {\bf H}_{\psi,k}^* {\bf H}_{\psi,k}^T {\bf T}_{\psi,-k} $. 
Furthermore, as ${\bf t}_{\psi,k,q}$ and ${\bf t}_{\psi,k,p}$ are orthogonal to each other for $p \neq q$, we have that
\begin{align}
   & \Big( {\bf t}_{\psi,k,p}^H  {\bf H}_{\psi,k}^T  {\bf T}_{\psi,-k} \Big) \Big(  {\bf T}_{\psi,-k}  {\bf H}_{\psi,k}^* {\bf t}_{\psi,k,q} \Big) \notag\\
   &= {\bf t}_{\psi,k,p}^H  {\bf H}_{\psi,k}^T  {\bf T}_{\psi,-k}  {\bf H}_{\psi,k}^* {\bf t}_{\psi,k,q} 
   = \lambda_{\psi,k,q} {\bf t}_{\psi,k,p}^H {\bf t}_{\psi,k,q} =0,\notag
\end{align}
which indicates that the eigenvectors ${\bf T}_{\psi,-k}  {\bf H}_{\psi,k}^* {\bf t}_{\psi,k,q}$ and ${\bf T}_{\psi,-k}  {\bf H}_{\psi,k}^* {\bf t}_{\psi,k,p}$ associated with different non-zero eigenvalues of $ {\bf T}_{\psi,-k} {\bf H}_{\psi,k}^* {\bf H}_{\psi,k}^T {\bf T}_{\psi,-k} $ are also orthogonal to each other. 
Therefore, the  precoding vector ${\bf v}_{\psi,k,q}$ for ${\rm U}_{\psi,k}$ in \eqref{V_mat_exp} can be rewritten as in \eqref{v_psikq_design_exp}.

Given the precoding vector ${\bf v}_{\psi,k,q}$ in \eqref{v_psikq_design_exp} and given an MRC receiver, the $\text{SINR}_{\psi,k,q}^{\text{BD-MRC}}$ in \eqref{SINR_BDMRC_def2} is derived as in \eqref{SINR_BD_MRC_eq_proof}, shown at the bottom of this page, where step $(a)$ follows from \eqref{eigen_vec_Omega}.
\begin{figure*}[b]
\vspace{-0.3cm}
\hrulefill
\begin{align}\label{SINR_BD_MRC_eq_proof}
\resizebox{0.94\hsize}{!}{$
    \text{SINR}_{\psi,k,q}^{\text{BD-MRC}} 
   \! =\! \frac{ ( {\bf t}_{\psi,k,q}^H {\bf H}_{\psi,k}^T  {\bf T}_{\psi,-k} ) {\bf T}_{\psi,-k} {\bf H}_{\psi,k}^* {\bf H}_{\psi,k}^T {\bf T}_{\psi,-k} ( {\bf T}_{\psi,-k} {\bf H}_{\psi,k}^* {\bf t}_{\psi,k,q} )}{||{\bf T}_{\psi,-k}  {\bf H}_{\psi,k}^* {\bf t}_{\psi,k,q}||^2   \frac{N_0}{P_{\psi,k,q}} }
   \overset{(a)}{=} \frac{ \lambda_{\psi,k} ( {\bf t}_{\psi,k,q}^H {\bf H}_{\psi,k}^T  {\bf T}_{\psi,-k} )  ( {\bf T}_{\psi,-k} {\bf H}_{\psi,k}^* {\bf t}_{\psi,k,q} )}{||{\bf T}_{\psi,-k}  {\bf H}_{\psi,k}^* {\bf t}_{\psi,k,q}||^2   \frac{N_0}{P_{\psi,k,q}} }
    $}
\end{align}
\end{figure*}
and which directly leads to \eqref{SINR_BD_MRC_lem}.

\section{Proof of Lemma \ref{BD_Lem1}}\label{Proof_Lem1}
For $L \to \infty$ and finite $M_{\psi,k}$, we can use the Trace Lemma (cf. \cite{Debbah}) to derive that
\begin{align}
    &  \frac{1}{L} \frac{{\bf h}_{\psi,k,\ell}^T }{\sqrt{\beta_{\psi,k}} }  {\bf T}_{\psi,-k} \frac{{\bf h}_{\psi,k,\ell}^*}{\sqrt{\beta_{\psi,k}} }	 \notag\\
    &\hspace{1cm}\stackrel{a.s.}{\longrightarrow} \frac{1}{L}  {\rm Tr}\Big\{ {\bf T}_{\psi,-k} \Big\} \simeq 1 - \frac{M_\psi - M_{\psi,k} }{L} ,  \\
    & \frac{1}{L}  {\bf h}_{\psi,k,\ell}^T  {\bf T}_{\psi,-k} {\bf h}_{\psi,k,\ell'}^*  	\stackrel{a.s.}{\longrightarrow} 0, \text{ for } \ell' \neq \ell,
\end{align}
where ${\bf h}_{\psi,k,\ell}$ denotes the $\ell$-th column of ${\bf H}_{\psi,k}$, and where $\stackrel{a.s.}{\longrightarrow}$ denotes the almost sure convergence.  In the above, we also consider that ${\rm Tr}\big\{ {\bf T}_{\psi,-k} \big\} = {\rm Rank}\big\{ {\bf T}_{\psi,-k} \big\}$ because ${\bf T}_{\psi,-k}$ is a projection matrix.  Therefore, we can now derive that
\begin{align}
    &\frac{1}{L} {\bf H}_{\psi,k}^T {\bf T}_{\psi,-k} {\bf H}_{\psi,k}^*  \simeq \beta_{\psi,k}  \Big( 1 \!-\! \frac{M_\psi - M_{\psi,k} }{L}  \Big) {\bf I}_{M_{\psi,k}}, 
\end{align}
which reveals that all eigenvalues of $\frac{1}{L} {\bf H}_{\psi,k}^T {\bf T}_{\psi,-k} {\bf H}_{\psi,k}^*$ become identical as $L \to \infty$. Therefore, we have the following asymptotic result
\begin{align}\label{lambda_approx_Massive}
    \lambda_{\psi,k,q} \simeq {\beta_{\psi,k}}  \left( L - M_\psi + M_{\psi,k}  \right).
\end{align}
Then, the effective rate in \eqref{Rate_Multiplex_Single} has the asymptotic behavior
\begin{align}
     R_{\psi,k} \simeq \xi_{G,Q} \sum\limits_{q=1}^{J_{\psi,k}}  \ln \bigg(\!  1+ \frac{P_{\psi,k,q}}{N_0}  {\beta_{\psi,k}}  \left( L - M_\psi + M_{\psi,k}  \right) \!\bigg).
\end{align}
According to the properties of the water-filling algorithm, we know that equal-power allocation among $\{s_{\psi,k,q}: q \in [J_{\psi,k}]\}$ is optimal, which in turn tells us that $R_{\psi,k}^\star(P_{\psi,k}) \simeq \mathring{R}_{\psi,k}^\star$ where $ \mathring{R}_{\psi,k}^\star$ is given by
\begin{align}
     \mathring{R}_{\psi,k}^\star \triangleq& \xi_{G,Q} J_{\psi,k}  \ln\! \bigg( \! 1+ \frac{P_{\psi,k}}{N_0 J_{\psi,k}} {\beta_{\psi,k}}  \left( L - M_\psi + M_{\psi,k}  \right) \bigg).
\end{align} 
At this point, we use the asymptotic equivalence $\mathring{R}_{\psi,k}^\star$ to substitute the objective function of the MMF optimization in \eqref{BD_MRC_Multiplex_Opt}, given by
\begin{align}\label{BD_MRC_Multiplex_Opt_LargeL}
\mathcal{S}_3 \begin{cases}
    &\max\nolimits_{\mathcal{P}_\Psi} \min\nolimits_{\psi \in \Psi} \min\nolimits_{k \in [Q]} \ \  \mathring{R}_{\psi,k}^\star \\
    & \text{s. t. } P_t = \sum\nolimits_{\psi \in \Psi} \sum\nolimits_{k \in [Q]}  P_{\psi,k} = P_{\rm tot}.
\end{cases}    
\end{align}
Finally, as all users have the same effective rate under optimal power allocation for \eqref{BD_MRC_Multiplex_Opt_LargeL}, we can use the total power constraint to derive \eqref{R_BD_MRC_LargeL_eq}--\eqref{R_BD_MRC_LargeL_M_eq} respectively.

\section{Proof of Proposition \ref{Prop_ZF_singe_rate_multi}}\label{Proof_Prop_ZF_single_rate}
For any $\psi \in \Psi$, $k \in [Q]$ and $q \in [M_{\psi,k}]$, let ${\bf e}_{\psi,k,q} \in \mathbb{C}^{M_{\psi} \times 1}$ denote a vector with all zero elements except  the $k(q)$-th element equalling 1, where  $k(q) \triangleq q + \sum_{k'=1}^{k-1} M_{\psi,k'}$.  The precoding vector for symbol $s_{\psi,k,q}$ under ZF precoding (cf. \eqref{signal_transmit_multi} and \eqref{ZF_Design_Multi}) can be written as
    \begin{align}
    {\bf v}_{\psi,k,q} =  {{\bf H}}^*_{\psi} \big( {{\bf H}_{\psi}^T} {{\bf H}_{\psi}^*}  \big)^{-1}  {\bf e}_{\psi,k,q} \sqrt{\big(\big[ \big( {\bf H}_{\psi}^T {\bf H}_{\psi}^* \big)^{-1} \big]_{k(q),k(q)} \big)^{-1}}.
    \end{align}
   After removing the inter-group interference and considering that ${\bf H}_{\psi}^T  {{\bf H}}^*_{\psi} \big( {{\bf H}_{\psi}^T} {{\bf H}_{\psi}^*}  \big)^{-1} = {\bf I}_{M_\psi}$, we  simplify the received signal for decoding $s_{\psi,k,q}$  in \eqref{receive_signal_multi}  under ZF precoding (cf.~\eqref{ZF_Design_Multi}) as
    \begin{align}
         y_{\psi,k,q}' \!=\! \sqrt{\big(\big[ \big( {\bf H}_{\psi}^T {\bf H}_{\psi}^* \big)^{-1} \big]_{k(q),k(q)} \big)^{-1}} \sqrt{{P}_{\psi,k,q}} {\bf s}_{\psi,k,q}  \!+ {z}_{\psi,k,q}'.   
    \end{align}
    Then, the effective average rate at ${\rm U}_{\psi,k}$ is of the form
    \begin{align}
        &\bar R_{\psi,k}^{\text{ZF}} \!=\! \xi_{G,Q}  \mathbb{E} \Bigg\{ \!\sum\limits_{q=1}^{M_{\psi,k}}  \ln\! \Bigg( 1 \!+\! \frac{P_{\psi,k,q} / N_0 }{ \big[ \big( {\bf H}_{\psi}^T {\bf H}_{\psi}^* \big)^{-1} \big]_{k(q),k(q)}} \Bigg) \!\Bigg\} \label{R_ZF_single_Excet}\\
        &\hspace{0.1cm}\overset{(a)}{\ge}   \xi_{G,Q} \sum\limits_{q=1}^{M_{\psi,k}} \ln \left( 1 + \frac{P_{\psi,k,q} / N_0}{ \mathbb{E} \Big\{  \big[ \big( {\bf H}_{\psi}^T {\bf H}_{\psi}^* \big)^{-1} \big]_{k(q),k(q)}  \Big\}}  \right),
    \end{align}
    where $(a)$ follows from using Jensen's inequality on the convex function $\ln(1+x^{-1})$. Considering that $\mathbb{E} \left\{  \big[ \big( {\bf H}_{\psi}^T {\bf H}_{\psi}^* \big)^{-1} \big]_{k(q),k(q)}  \right\} = \frac{1}{\beta_{\psi,k}(L-M_\psi)}$ (cf. \cite{Wong_Netw}), we obtain the lower-bound in  \eqref{lower_ZF_multi}.
    
    To obtain the upper-bound of $\bar R_{\psi,k}^{\text{ZF}}$, we re-use Jensen's inequality for the concave function $\ln(1+x)$ in \eqref{R_ZF_single_Excet}, and obtain 
    \begin{align}
        &\bar R_{\psi,k}^{\text{ZF}}  \notag\\
        &\le \xi_{G,Q}\! \sum\limits_{q=1}^{M_{\psi,k}}   \ln\! \Bigg( 1 + \frac{P_{\psi,k,q}}{N_0} \mathbb{E} \Bigg\{ \frac{1}{ \big[ \big( {\bf H}_{\psi}^T {\bf H}_{\psi}^* \big)^{-1} \big]_{k(q),k(q)}} \Bigg\} \Bigg), 
    \end{align}
    which induces the upper-bound in \eqref{upper_ZF_multi} by considering that $\mathbb{E}\big\{ \big( { \big[ \big( {\bf H}_{\psi}^T {\bf H}_{\psi}^* \big)^{-1} \big]_{k(q),k(q)}} \big)^{-1}
    \big\} = \beta_{\psi,k}(L-M_\psi+1)$ (cf. \cite{Wong_Netw}).

	\bibliographystyle{IEEEtran}				
	\bibliography{IEEEabrv,aBiblio}			

\end{document}